\documentclass[a4paper,11pt]{article}
\usepackage{setspace}
\usepackage[english]{babel}
\usepackage[T1]{fontenc}
\usepackage{lineno}
\usepackage[latin1]{inputenc}
\usepackage{amsmath}
\usepackage{amssymb}
\usepackage{amsfonts}
\usepackage{graphicx}
\usepackage{textcomp}
\usepackage{geometry}
\usepackage{eufrak}
\usepackage{url}
\usepackage{mathrsfs}

\DeclareMathAlphabet{\mathpzc}{OT1}{pzc}{m}{it}

\newtheorem{theorem}{Theorem}[section]
\newtheorem{lemma}[theorem]{Lemma}
\newtheorem{proposition}[theorem]{Proposition}

\newtheorem{assumption}[theorem]{Assumption}
\newtheorem{result}[theorem]{Result}

\newenvironment{proof}[1][Proof.]{\begin{trivlist}
\item[\hskip \labelsep {\bfseries #1}]}{\end{trivlist}}
\newenvironment{definition}[1][Definition.]{\begin{trivlist}
\item[\hskip \labelsep {\bfseries #1}]}{\end{trivlist}}

\newenvironment{remark}[1][Remark.]{\begin{trivlist}
\item[\hskip \labelsep {\bfseries #1}]}{\end{trivlist}}

\newcommand{\qed}{\nobreak \ifvmode \relax \else
      \ifdim\lastskip<1.5em \hskip-\lastskip
      \hskip1.5em plus0em minus0.5em \fi \nobreak
      \vrule height0.75em width0.5em depth0.25em\fi}



\begin{document}

\title{Is Quantum Field Theory ontologically interpretable? \\ On localization, particles and fields in\\ relativistic Quantum Theory}
\author{David Schroeren\\ Wolfson College, University of Cambridge\footnote{This essay was submitted as part of required work for the Master of Advanced Study in Mathematics (Part III of the Mathematical Tripos).}\\
\footnotesize Supervisor: Dr. Jeremy Butterfield}

\onehalfspacing
\maketitle
\thispagestyle{empty}
\newpage
\tableofcontents
\newpage
\section{Introduction}
Quantum field theory (QFT) is ubiquitous in particle physics and condensed matter physics. Most modern theories in quantum physics, including the Standard Model, BCS theory and superconductivity, are formulated as quantum field theories - which have been found to provide predictions of unparalleled accuracy. While the success of quantum field theory surpasses even that of standard quantum mechanics, the latter has been subject to intense scrutiny over more than half a century by physicists, philosophers and mathematicians alike. Although, in the eyes of scientific realists, the foundational problems of fixed-number quantum mechanics are by no means resolved, its mathematical framework is largely understood. This is not true for QFT: many concepts used in heuristic QFT are mathematically ill-defined, while the well-defined framework of algebraic QFT is not straighforwardly extendible to interacting systems. Accordingly, the task of interpreting relativistic QFT is somewhat complicated. There are two ways that have been put forward to provide an ontology for QFT. Some have suggested that, in accordance with our everyday sense-experience of localized macroscopic objects and events, \emph{particles} should be fundamental to the theory. However, others favour a \emph{field} ontology for quantum \emph{field} theory. And since both views face difficulties, it has questioned whether QFT is ontologically interpretable at all. In this essay, I attempt an answer to this question, examining ways in which one might be able to equip QFT with an explicit realist ontology by taking a closer look at the ontological status of particles, localization and fields in QFT.\\ 
\\
I shall proceed as follows: in Section 1, I will start by setting up the algebraic framework used to describe the foundations of non-interacting QFT and elucidate its relation to the physical content of the familiar Hilbert-space formalism. In Section 2, I will provide a formal set of assumptions about localization and give a criterion by which to identify whether or not a particle concept is \emph{fundamental} to QFT. Then, I will construct a na\"{i}ve approach to localization for the free bosonic case (Section 3) and show how this localization scheme, as a consequence of the Reeh-Schlieder theorem, fails to satisfy this \emph{fundamentality} criterion (Section 3). This will lead to the Newton-Wigner localization scheme, which I report in Section 4. The NW concept of localization, at first sight, seems to alleviate the problems of na\"{i}ve localization. However, as I will show, NW style localization only superficially circumvents the Reeh-Schlieder theorem, in that it fails to obey strong microcausality. I will then examine attempts to generalise Newton-Wigner localization in such a way that it pertains to a localized property of a system, such as the center of energy (CE) or the center of mass (CM) and argue that they are not satisfactory, which will lead to my intermediate conclusion: QFT does not permit an ontology contingent on the localization of a particle or a property of a system in regions of space or spacetime. In Section 5, I will turn my attention to fields and explain how particles can be regarded as an emergent phenomenon of a relativistic field theory. Finally, I shall conclude that since QFT does not lend itself to a particle ontology, we need to interpret QFT in terms of \emph{fundamental fields} and regard particle observation events as emergent phenomena. 

\section{Mathematical Preliminaries}\label{mathprelims}
In this Section, I will provide a minimal set of mathematical tools necessary for the algebraic approach used throughout this essay. To construct algebraic quantum field theory, I will begin by introducing abstract algebras.\footnote{Following Halvorson in \cite{HalvorsonMueger2007}.}
\begin{definition}
A \emph{C*-algebra} is a pair consisting of a *-algebra $\mathfrak{A}$ and a norm $\|\cdot \|: \mathfrak{A} \rightarrow \mathbb{R}$ such that $\mathfrak{A}$ is complete with respect to $\|\cdot \|$ and
\[\|AB\| \leq \|A\|\|B\|,\qquad \qquad \|A^*A\| = \|A\|^2\]
for all $A, B \in \mathfrak{A}$. Let $\mathfrak{A}$ denote the algebra equipped with its norm. Also, I will only consider C*-algebras that contain the multiplicative identity $I$, viz. \emph{unital} C*-algebras.
\end{definition}
In other words: A C*-algebra $\mathfrak{A}$ is a Banach-algebra  over the complex numbers, together with an involution $*: \mathfrak{A} \rightarrow \mathfrak{A}$, which takes the adjoint of elements in the algebra. In any quantum theory, we would like to evaluate an operator with respect to some quantum system. Thus, we define a state on an algebra.
\begin{definition}
A \emph{state} on an algebra $\mathfrak{A}$ is a linear functional $\omega: \mathfrak{A} \rightarrow \mathbb{C}$ such that 
\begin{enumerate}
	\item $\omega(A^*A) \geq 0, \forall A \in \mathfrak{A}$
	\item $\omega(I) = 1.$
\end{enumerate}
\end{definition}
\begin{definition}
A state $\omega$ on $\mathfrak{A}$ is \emph{mixed} if $\omega$ is a mixture of at least two other states, i.e. if $\omega$ can be expressed as  $\omega = \lambda \omega_1 + (1-\lambda_2)\omega_2$ for $\omega_1 \neq \omega_2$ and $\lambda \in (0,1)$. Otherwise, $\omega$ is \emph{pure}.
\end{definition}
In other words, a state is mixed if it is a nontrivial convex combination of other states on $\mathfrak{A}$.\\
Note that I have not talked about any Hilbert space yet! In a few moments, we shall see how this relates to the familiar notion of a statevector in a Hilbert space. But first, let me give a few more essential definitions that might help readers familiar with the usual approach to QFT make more sense of C*-algebras.
\begin{definition}
Let $\mathcal{H}$ be a Hilbert space. Let $\mathfrak{B}(\mathcal{H})$ be the C*-algebra of bounded linear operators on $\mathcal{H}$, i.e. $\mathfrak{B}(\mathcal{H})$ is a C*-algebra satisfying
\begin{enumerate}
	\item $\forall A \in \mathfrak{B}(\mathcal{H})$ $\exists$ a smallest number $\|A\|$ such that $\langle Ax, Ax \rangle^{1/2} \leq \|A\|$ for every unit vector $x \in \mathcal{H}$
	\item $\forall A, B \in \mathfrak{B}(\mathcal{H})$ let $AB$ denote the binary composition of two elements of $\mathfrak{B}(\mathcal{H})$
	\item $\forall A \in \mathfrak{B}(\mathcal{H})$ let $A^*$ denote the unique element such that $\langle A^*x,y \rangle = \langle x,Ay\rangle $ for all $x, y \in \mathcal{H}$
\end{enumerate}
\end{definition}
In particular, we observe that $\mathfrak{B}(\mathcal{H})$ contains the observables (i.e. the self-adjoint operators) as a subset.\\
In the context of localization, I will be talking about subalgebras of $\mathfrak{B}(\mathcal{H})$. We would like these subalgebras to obey a number of nice properties as given by the following:
\begin{definition}
Let $\mathfrak{R}$ be a *-subalgebra of $\mathfrak{B}(\mathcal{H})$. Then $\mathfrak{R}$ is a \emph{von Neumann algebra} if
\begin{enumerate}
	\item $I \in \mathfrak{R}$
	\item $(\mathfrak{R}')' = \mathfrak{R}$
\end{enumerate}
where $\mathfrak{R}' = \lbrace A\in \mathfrak{B}(\mathcal{H}): [A,C]=0,  \forall C\in \mathfrak{R}\rbrace$ is the \emph{commutant} of $\mathfrak{R}$.
\end{definition}
\begin{remark}
According to von Neumann's \emph{double commutant theorem}, condition (2) in the above definition is equivalent to the statement that the unital *-algebra $\mathfrak{R}$ is weakly closed, i.e. $\mathfrak{R}$ is closed in the weak operator topology.\footnote{The weak operator topology is defined in terms of a family $\{ \langle u, Av\rangle : u, v \in \mathcal{H}\}$ so that a net $\{A_i\}_{i\in\mathcal{I}}$ converges \emph{weakly} to $A$ if the sequence $\{\langle u, A_i v\rangle \}_{i\in \mathcal{I}}$ converges to $\langle u, Av\rangle$ for all $u,v \in \mathcal{H}$; see \cite{HalvorsonMueger2007} for details.}
\end{remark}
How do we connect an abstract C*-algebra to the algebra of bounded linear operators on a Hilbert space? 
\begin{definition}
Let $\mathfrak{A}$ be a C*-algebra. A \emph{representation} of $\mathfrak{A}$ is a pair $(\pi,\mathcal{H_{\pi}})$ where $\mathcal{H_{\pi}}$ is a Hilbert space and $\pi$ is a homomorphism from $\mathfrak{A}$ into $\mathfrak{B}(\mathcal{H_{\pi}})$ that obeys $\pi(A^*) = \pi^*(A), \forall A\in \mathfrak{A}$. Additionally, a representation $(\pi, \mathcal{H_{\pi}})$ is said to be 
\begin{enumerate}
	\item \emph{irreducible} if $\pi(\mathfrak{A})$ leaves no subspace of $\mathcal{H_{\pi}}$ invariant;
	\item \emph{faithful} if $\pi$ is an isomorphism;
	\item \emph{factorial} if $\pi(\mathfrak{A})$ is a \emph{factor}, i.e. if the centre of $\pi(\mathfrak{A})$ contains only multiples of the identity.
\end{enumerate}
\end{definition}
\begin{remark}
In general, $\pi(\mathfrak{A}) \neq \pi(\mathfrak{A})''$. If this is the case however, it follows that $\pi(\mathfrak{A})''$ definitely \emph{is} a von Neumann algebra and we observe the following: the irreducibility of a representation $(\pi, \mathcal{H_{\pi}})$ is equivalent to the statement that $\pi(\mathfrak{A})'' = \mathfrak{B}(\mathcal{H_{\pi}})$. Since $\mathfrak{B}(\mathcal{H_{\pi}})$ is a factor, we can deduce that if a representation is irreducible, it is also factorial.
\end{remark}
\begin{definition}
Let $(\pi, \mathcal{H_{\pi}})$ and $(\phi, \mathcal{H_{\phi}})$ be representations of a C*-algebra $\mathfrak{A}$. $(\pi, \mathcal{H_{\pi}})$ and $(\phi, \mathcal{H_{\phi}})$ are said to be
\begin{enumerate}
	\item \emph{unitarily equivalent} if there exists a unitary mapping $U:\mathcal{H_{\pi}} \rightarrow \mathcal{H_{\phi}}$ such that \\ 
	$U^{-1}\phi(A)U = \pi(A)$ for all $A\in \mathfrak{A}$;
	\item \emph{quasiequivalent} if there exists a *-isomorphism between the von Neumann algebras $\pi(\mathfrak{A})''$ and $\phi(\mathfrak{A})''$;
	\item \emph{disjoint} if they are not quasiequivalent.
\end{enumerate}
\end{definition}
Clearly, if two representations are unitarily equivalent, they are also quasiequivalent.
\begin{definition}
Let $\mathfrak{R}\subseteq \mathfrak{B}(\mathcal{H})$. A vector $x\in \mathcal{H}$ is \emph{cyclic} with respect to $\mathfrak{R}$ if the closed linear span of the set $\{ Ax:A\in \mathfrak{R}\}$ is the whole of $\mathcal{H}$, that is
\begin{equation}\text{span}_{\mathbb{C}}\overline{\{ Ax:A\in \mathfrak{R}\}} = \mathcal{H}.
\end{equation}
A vector $x\in \mathcal{H}$ is said to be \emph{separating} for $\mathfrak{R}$ if $Ax=0$ and $A\in \mathfrak{R}$ entails $A=0$.
\end{definition}
This just means that if a vector x is cyclic with respect to a subset $\mathfrak{R}$ of $\mathfrak{B}(\mathcal{H})$, we can span the whole Hilbert space $\mathcal{H}$ by actions of $\mathfrak{R}$ on x. Trivially, x is a cyclic vector for $\mathfrak{B}(\mathcal{H})$, the algebra of bounded linear operators on $\mathcal{H}$.\\
At last, we can set out to relate our state on an algebra to our familiar Hilbert space vector formalism via the Gelfand-Naimark-Segal (GNS) theorem. It shows that every C*-algebraic state can be represented by a vector in a Hilbert space.
\begin{theorem} \emph{(GNS)}\\
Let $\omega$ be a state on $\mathfrak{A}$. Then there exists a representation $(\pi, \mathcal{H_{\pi}})$ of $\mathfrak{A}$, and a unit vector $x\in \mathcal{H}$ such that
\begin{enumerate}
	\item $x$ is cyclic with respect to $\pi(\mathfrak{A})$;
	\item $\omega(A) = \langle x, \pi(A)x \rangle $, for all $A\in \mathfrak{A}$.
\end{enumerate}
Additionally, the representation $(\pi, \mathcal{H_{\pi}})$ satisfying (1) and (2) is unique up to unitary equivalence.
\label{gns}
\end{theorem}
So the C*-algebraic state is just the expectation value of the familiar Hilbert space formalism. Note that, in particular, we can represent the C*-algebra $\mathfrak{B}(\mathcal{H})$ of bounded linear operators on some Hilbert space $\mathcal{H}$ unto itself. \\
The advantage of the algebraic approach is that, provided the GNS construction, the state on an algebra does not depend on our choice of representation - even when confronted with unitarily \emph{in}equivalent representations. In the algebraic setting, the entire physical content of the theory is contained in the \emph{algebra} of operators. In the next Section, we shall see how to give life to this idea.\\

\section{Localization and the particle ontology} \label{locpartont}
With the algebraic framework in place, I will now set out to formulate what is generally meant by `localization' and give a criterion by which we can then evaluate if a relativistic quantum theory permits a particle ontology. 
\subsection{Basic assumptions about localization}
What does it mean for something to be localized? One can interpret localization both in the context of space and spacetime.\footnote{Pointlike localization could then be regarded as the limit of arbitrarily small regions.} In this essay, the latter will generally be assumed to obey \emph{global hyperbolicity}. That is, the spacetime is assumed to have a topology $\Sigma \times \mathbb{R}$ where $\Sigma$ is some Cauchy surface.\footnote{See theorem 4.1.1 (Geroch (1970); Dieckmann (1988)) in \cite{Wald1994}, p.56.} Minkowski space, which will be used for the most part of this essay, is trivially an example of a globally hyperbolic spacetime. Let me now introduce a number of general concepts for localization in Minkowski spacetime. In the following, I will loosely refer to local measurements that can be made in a region; then in subsection \ref{fundcrit}, I will give a more formal definition of local measurability.
\begin{definition}
Let $\mathcal{O}$ be the collection of bounded open subsets of Minkowski space $M$ and let $\mathcal{G}$ denote the translation group on $M$. A \emph{net} of local algebras over $M$ is a set
\begin{equation}
 \mathfrak{A}=\{\mathfrak{R}(O): O\in \mathcal{O}\}
\end{equation}
where every $\mathfrak{R}(O)$ is a von Neumann subalgebra of $\mathfrak{B}(\mathcal{H})$. Furthermore, if the net $\mathfrak{A}$ is generated by the set $\{\mathfrak{R}(O+\mathbf{x}): \mathbf{x}\in \mathcal{G}\}$ for a particular $O\subseteq M$, the net is said to obey \emph{additivity}.\footnote{In this essay, I shall use the following notation: curly lowercase letters such as $x\in M$ will denote a point in Minkowski space, while $\mathbf{x} \in \mathcal{G}$ is an element of the translation group on $M$.}
\end{definition}
\begin{remark}
The physical idea will be that a self-adjoint $A \in \mathfrak{R}(O)$ is an observable that can be measured by an operational procedure confined to $O$. I will refer to such a measurement procedure as \emph{local in $O$}.
\end{remark}
\begin{remark}
Since $M$ is a topological space, we can always find a covering $\mathcal{O} \subseteq \mathcal{T}$ of bounded open subsets of $M$ with respect to a suitable topology $\mathcal{T}$. In particular, denote the set of open \emph{double cones} in Minkowski space by
\begin{equation}
\mathcal{U} = \{D^+(x)\cap D^-(y): \quad x,y\in M\text{ and } y\in D^+(x)\},
\label{eq:doublecone}
\end{equation} 
where $D^{\pm}(x)$ is the \emph{causal future/past} of $x$. Sometimes, it will be convenient to work with double cones as a specific example of bounded open subsets of Minkowski space.
\end{remark}
Similarly, one can define a net of local algebras over $\mathbb{R}^3$ by $ \{ \mathfrak{R}(G+\mathbf{x}): G\subseteq \mathbb{R}^3 \text{ and } \textbf{x} \in \mathcal{G} \mid_{\mathbb{R}^3}\} $, which is additive with respect to translations in $\mathbb{R}^3$ for some particular choice of a single set $G\subseteq \mathbb{R}^3$.\\
\\
Now that we have an algebra local in a region $O$ of Minkowski space,\footnote{In this essay, $O$ will always denote to a subset of Minkowski space, while $G$ will always refer to a spatial set.} we can start to think about basic assumptions that any localization scheme $O\rightarrow \mathfrak{R}(O)$ would be expected to satisfy.
\begin{assumption}
\emph{(Isotony)}\\
Let $O_1, O_2$ be double cones in Minkowski space. A net $\mathfrak{A}$ of local algebras satisfies isotony just in case $O_1 \subseteq O_2$ then $\mathfrak{R}(O_1) \subseteq \mathfrak{R}(O_2)$, for $\mathfrak{R}(O_1), \mathfrak{R}(O_2)\in \mathfrak{A}$.
\label{isotony}
\end{assumption}
This is just about as innocuous as it looks. If a measurement procedure is local in a region $O_1$ that is a subset of another region $O_2$, then the measurement procedure should also be local in $O_2$.\\
Before formulating a covariance requirement for a net of local algebras in Minkowski space, let us consider the special case of a spatial net $G \rightarrow \mathfrak{R}(G)$ where $G \subseteq \mathbb{R}^3$:
\begin{assumption}
\emph{(Translation-Rotation Covariance)}\\
Let $G \subseteq \mathbb{R}^3$ be a bounded open set. A net $\mathfrak{A}$ of local algebras is said to obey translation-rotation covariance if, for all $\mathfrak{R}\in\mathfrak{A}$ and for all elements $U(R,\textbf{a})$ of a unitary representation of the Euclidean group $SO(3) \ltimes \mathbb{R}^3$ with elements $(R,\textbf{a})$,
\begin{equation}
	U^{-1}(R,\textbf{a})\mathfrak{R}(G)U(R,\textbf{a}) = \mathfrak{R}(RG+\textbf{a})
\end{equation}
where $R \in SO(3)$ are rotations about the origin and $\textbf{a}$ is a translation in $\mathbb{R}^3$.
\label{trcovariance}
\end{assumption}
In other words: A local algebra translated and rotated about the origin should be the same as an algebra local in the translated and rotated spatial set. So if we were to implement a local measurement procedure by a local projector $P_G$ that yields an eigenvalue 1 if (say) a particle is localized in $G$ and $0$ if it is not, then actions of the Euclidean group will yield a projector $P_{G'}$ such that $P_{G'}\equiv 0$ (if $G\cap G' = \emptyset$). In terms of the Hilbert space formalism: If a statevector is said to be localized in a spatial set $G$, it should be orthogonal to a statevector localized in a disjoint spatial set $G'$.\\
\\
How can we generalize this to a net $O \mapsto \mathfrak{R}(O)$, where $O$ is a double cone in Minkowski space? Since Minkowski spacetimme is an \emph{affine} space, we cannot straightforwardly use the concept of rotations about the origin. An intuitive way to accomplish this would  be to require \emph{Poincaré} covariance, viz. the covariance of any local algebra under actions of the inhomogenous Lorentz group $\mathfrak{P}$ of Lorentz transformations and spacetime translations. However, for our purposes, the following weaker assumption is sufficient:\footnote{In fact, Poincaré covariance of observables together with strong microcausality, the existence of a unitary representation of the covering group $\overline{\mathfrak{P}}$ of the Poincaré group $\mathfrak{P}$ and a unique Poincaré-invariant state on $\mathcal{H}$ as well as the spectrum condition (see \ref{spectrum}) form a subset of the Wightman axioms which are the basis of \emph{Axiomatic QFT}; cf. \cite{Haag} pp. 58.}
\begin{assumption}\emph{(Translation Covariance)}\\
Let $O \subseteq M$ be a double-cone and let $\mathcal{G}$ denote the spacetime translation group on $M$. A net $\mathfrak{A}$ of local algebras is said to obey \emph{translation covariance} if there exists a faithful and continuous representation $\mathbf{x}\rightarrow \alpha_{\mathbf{x}}$ of $\mathcal{G}$ in the automorphism group on $\mathfrak{A}$ such that 
\begin{equation}
	\alpha_{\mathbf{x}}\mathfrak{R}(O) = \mathfrak{R}(O+ \mathbf{x})
\end{equation}
for all $\mathfrak{R}\in\mathfrak{A}$ and $\mathbf{x} \in \mathcal{G}$.
\label{ttcovariance}
\end{assumption}
Now let us turn to the central relativistic assumption: No local operation may affect the outcome of another local operation at spacelike separation. In other words, it should be impossible to observe act-outcome correlations at spacelike separation. This seems appropriate since quantum field theory is constructed as a Lorentz-covariant, viz. \emph{relativistic} theory. 
\begin{assumption}
\emph{(Strong Microcausality)}\\
A net $\mathfrak{A}$ of local algebras is said to satisfy strong microcausality if, for every pair $O_1, O_2$ of spacelike-separated double cones in Minkowski space, the commutator of two algebra elements, one local in $O_1$ and the other local in $O_2$, vanishes:
\begin{equation}
[\mathfrak{R}(O_1), \mathfrak{R}(O_2)] = \{0\}.
\end{equation}
where $\mathfrak{R}(O_1),\mathfrak{R}(O_2)\in \mathfrak{A}$. 
\label{sm}
\end{assumption}
\begin{remark}
Let $O$ be a bounded open set in Minkowski space and let $D(O) = D^+(O)\cup D^-(O)$ denote the \emph{domain of dependence}\footnote{Definition extracted from Muller and Butterfield in \cite{MB1994}, p. 460} of $O$, where $D^{\pm}(O) = \bigcup_{x\in O} D^{\pm}(x)$. Then, \ref{sm} is equivalent to 
\begin{equation}
\mathfrak{R}(O') \subseteq \mathfrak{R}(O)',
\label{eq:altsm}
\end{equation}where $O' = M \setminus D(O)$ denotes the causal complement of $O$.
\end{remark}
For spatial algebras, this can be put as follows: Given a foliation of Minkowski space into a family $\Sigma_t$ of spacelike hyperplanes and two spatial sets G, G' on any hyperplane. Then, a projector $P_G$ will commute with the projector $P_{G'}$ if G and G' are disjoint, and therefore spacelike separated.\\
In this essay, I will be considering spatial localization schemes. In these cases, we can replace assumption \ref{sm} by a weaker condition:
\begin{assumption}
\emph{(Weak Microcausality)}\\
Let $\Sigma_t$ be a spacelike slice of Minkowski space at some fixed time t. A net $\mathfrak{A}= \{\mathfrak{R}(G): G\subseteq \Sigma_t \}$ of local algebras is said to satisfy weak microcausality if, for every pair $G_1, G_2 \subseteq \Sigma_t $ of spacelike separated spatial sets on $\Sigma_t$, 
\begin{equation}
[\mathfrak{R}(G_1), \mathfrak{R}(G_2)] = \{0\}
\end{equation}
\label{wm}
where $\mathfrak{R}(G_1),\mathfrak{R}(G_2)\in \mathfrak{A}$.
\end{assumption}
The latter is sometimes called ``equal-time microcausality'', while the former is occasionally referred to as ``generalised microcausality''. \\
So we have drawn up a list of criteria that any physically plausible localization scheme - either pertaining to spatial sets or regions in spacetime -  should satisfy. With this list in hand, we can think about how to identify whether or not a theory permits a particle ontology.\\
\subsection{A \emph{fundamentality} criterion for particles}
\label{fundcrit}
In the previous Section, I have used ``measurable'' and ``observable'' rather informally. For the sake of clarity, I shall give a brief description of what I mean by \emph{measurement} in the context of an entirely \emph{non-interacting} theory. First of all, we observe the following definition which I will use to define local measurability.
\begin{definition}
Let $A$ denote an unbounded operator on some Hilbert space $\mathcal{H}$ and let $\mathfrak{R}(O) \subseteq \mathfrak{B}(\mathcal{H})$ be a von Neumann algebra local in a double cone $O$ where $O\subseteq M$. $A$ is said to be \emph{affiliated} with $\mathfrak{R}(O)$ if $U^{-1}AU = A$ for any unitary $U \in \mathfrak{R}(O)'$. 
\end{definition}
The following assumption will be instrumental in the formulation of the fundamentality criterion:
\begin{assumption} \emph{(Local measurability)}\\
Let $\mathfrak{R}(O)$ be an algebra local in a double cone $O\subseteq M$, in the sense that it satisfies assumptions (\ref{isotony}), (\ref{ttcovariance}) and (\ref{sm}). A physical quantity $A$ is called \emph{measurable in O} just in case it can be implemented by local operations contained in, or affiliated with, the algebra $\mathfrak{R}(O)$. A \emph{measurement outcome} for such a quantity is instantiated with an expectation given by the value of the state $\omega$ on $A$.
\end{assumption}
I would like to point out that this does not presuppose any outside observer but treats measurement as intrinsic to the theory itself. Also, this assumption should be not be regarded as a statement about the measurement problem in nonrelativistic quantum mechanics, which I have no intention to solve here.\\
From the above, we can deduce a notion of measurements local in subsets of $\mathbb{R}^3$. Let $O$ be a bounded open subset of Minkowski space with domain of dependence $D(O)$. Global hyperbolicity of Minkowski spacetime suggests that outcomes of operations local in a bounded open set $\tilde{O} \subseteq D^+(O)$ are \emph{determined} by operations local in $O$. In particular, if we let $\mathfrak{R}(O)$ denote the algebra local in $O \subseteq M$, then we can find a spatial set $G$  such that $O \subseteq D(G)$. This fact is known as the \emph{diamond axiom} and can be transcribed to any globally hyperbolic spacetime. Therefore, a physical quantity $A$ is defined to be \emph{measureable in a spatial set} $G \in \mathbb{R}^3$ if and only if $A$ is measurable in a bounded open set $O \subseteq D(G)$. Furthermore, if the diamond axiom holds (as is the case for Minkowski space), strong microcausality entails weak microcausality.
\\
Let me now proceed to spell out a necessary\footnote{By including the requirement that inertial and accelerating observers agree about the particle number in a given region, one could construct a sufficient condition for the existence of particles; cf. Arageorgis et al. in \cite{Arageorgis2003} and Baker in \cite{Baker2008}.} condition for the existence of a particle ontology for QFT: 
\begin{assumption} \emph{(Fundamentality criterion)}\\
Let $\mathfrak{A} = \{\mathfrak{R}(O): \mathfrak{R}(O) \subseteq \mathfrak{B}(\mathcal{H}) \text{ and } O\subseteq M\}$ be a net of local algebras for Minkowski space that satisfies assumptions (\ref{isotony}), (\ref{ttcovariance}) and (\ref{sm}). The particle concept is \emph{fundamental} to the QFT associated with $\mathfrak{B}(\mathcal{H})$ if and only if all physical quantities associated with the notion of a particle are locally measurable with respect to $\mathfrak{A}$.
\end{assumption}
Observe that, according to this criterion, any localization scheme that fails to satisfy assumptions (\ref{isotony}-\ref{wm}) cannot be used to give a particle ontology. Therefore, these assumptions also form a (partly redundant) set of \emph{necessary} conditions for a net $\mathfrak{A}$ to be interpreted as a particle localization scheme. \\
The central physical quantity associated with the notion of a particle is the \emph{number} of particles in a region of space or spacetime and I will examine several localization schemes with respect to this physical quantity.\\

\section{A na\"{i}ve localization scheme for quanta of the Klein-Gordon field}\label{naiveloc}
Given the above criteria, in this Section I will construct a localization scheme for the Klein-Gordon (KG) field and evaluate its conformity with these requirements. First, I will give a version of the \emph{first quantization} of the classical field that might appear to be somewhat different from the usual approach, before constructing a Fock space for the theory. Finally, I will show that thusly constructed `standard' localization scheme fails the \emph{fundamentality} criterion for the existence of a particle ontology.\\ 
Let us start with the action of the Klein Gordon theory,\footnote{For a thorough exposition of free bosonic QFT, see \cite{Peskin} chapter 2.} given by
\begin{equation}
	S = \int{d^4x(\partial \varphi\cdot\partial \varphi + m^2\varphi^2)}
\end{equation}
which, in flat Minkowski spacetime, can be written as
\begin{equation}
	S = \int{dt\mathcal{L}}
\end{equation}
where $\mathcal{L}$ is the Lagrangian density, given by
\begin{equation}
	\mathcal{L} = \frac{1}{2}\int{d^3x[(\dot{\varphi})^2 - (\nabla \varphi)^2-m^2\varphi^2]}
\end{equation}
The set of classical dynamical equations arising from the Klein Gordon action is the \emph{Klein Gordon equation}
\begin{equation}
 (-\Box + m^2)\varphi = 0.
\label{eq:KGequation}
\end{equation}
In classical mechanics, one is confronted with dynamical equations satisfied by functions that live on a finite-dimensional manifold. In the usual approach, we start out with a configuration space $Q$ (which is a smooth manifold) and construct the phase space of solutions to the dynamical equation as the cotangent bundle $T^*Q$, where the dynamics of the system are then implemented by the Hamiltonian vector field generated by a smooth real-valued function on $Q$. In the present context however, it is not straighforward to talk about a smooth manifold and its cotangent bundle, since the space of functions satisfying equation (\ref{eq:KGequation}) is \emph{infinite-dimensional}. Therefore, we cannot rely on usual methods. In the following, I will sketch the construction of an infinite-dimensional Hamiltonian system for the solution space of the KG field at a fixed time $t$.\\
Let $\mathscr{S}:=\{ f \in C^\infty(\mathbb{R}^n):  \|f\|_{\alpha,\beta} < \infty\, \forall \, \alpha, \beta \}$\ denote the Schwarz space\footnote{As defined in \cite{Conway1990}, p. 336.} of rapidly decreasing functions on $\mathbb{R}^3$. Furthermore, let
\begin{equation}
\|f\|_{\alpha,\beta} = \sup_{x\in\mathbb{R}^3}|x^\alpha D^\beta f(x)|
\label{eq:schwartznorm}
\end{equation}
denote the norm on $\mathscr{S}$, where $\alpha$ and $\beta$ are multi-indices. Then, it is straightforward to check that $\mathscr{S}$ is complete and normed with respect to $\|\cdot \|_{\alpha,\beta}$ and thus forms a \emph{Banach space}. Now, we observe the following definition.
\begin{definition}
A \emph{smooth Banach manifold} is a topological manifold in the sense that it is locally homeomorphic to a Banach space, together with a smooth structure which is given by an equivalence class of $C^{\infty}$-compatible\footnote{Here, a transition function is smooth with respect to the \emph{Fréchet derivative.}} atlasses.
\end{definition}
The norm (\ref{eq:schwartznorm}) induces a topology on $\mathscr{S}$ such that $\mathscr{S}$ is second countable and Hausdorff; furthermore, $\mathscr{S}$ is endowed with a natural smooth structure, given by the identity mapping. Therefore, we can regard $\mathscr{S}$ as an infinite-dimensional smooth manifold. Following \cite{Schmid}, let $\mathscr{S}^*$ be the dual space with respect to the $L^2$ norm $\langle \cdot, \cdot \rangle$ on $\mathscr{S}\times \mathscr{S}*$.\footnote{For the Klein Gordon field, in \cite{Schmid} Schmid refers to $\mathscr{S}^*$ as $\text{Den}(\mathbb{R}^3)$, the space of \emph{density functions} on $\mathbb{R}^3$.} The cotangent bundle of $\mathscr{S}$ is just given by $\mathcal{Q}:=T^*\mathscr{S}$ which is locally isomorphic to $\mathscr{S}\times \mathscr{S}^*$.\\
 By Darboux's theorem, any cotangent bundle is locally isomorphic to a symplectic manifold. $\mathcal{Q}$ can thus be interpreted as a symplectic structure and we can find a symplectic form $\sigma: \mathcal{Q}\times \mathcal{Q} \rightarrow \mathbb{R}$. It is straightforward to check that for every $\varphi_0 \oplus \pi_0, \varphi_1 \oplus \pi_1 \in \mathcal{Q}$, 
\begin{equation}
\sigma(\varphi_0 \oplus \pi_0,\varphi_1 \oplus \pi_1) = \langle \varphi_0 ,\pi_1\rangle - \langle \pi_0, \varphi_1\rangle
\label{eq:symp}
\end{equation}
is the bilinear, nondegenerate and antisymmetric, viz. \emph{symplectic} form on $\mathcal{Q}$. Here is another way to see this: because the KG equation (\ref{eq:KGequation}) is a second order partial differential equation, for a fixed time any solution is specified uniquely by its \emph{Cauchy data}, i.e. by a pair $(\varphi,\psi)$, where $\varphi$ represents the values of the function and $\psi$ those of its first derivative with respect to time. Furthermore, we know that the conserved current of the Klein Gordon theory is just given by
\begin{equation}
j= \varphi \partial^{\mu} \psi  - \psi\partial^{\mu} \varphi,\quad (\psi, \varphi) \in \mathcal{Q}
\end{equation}
For a fixed time, consider the time-component of this conserved vector current $j^0$:
\[j^0 = \varphi \partial^0 \psi - \psi \partial^0 \varphi.\]
Spatial integration yields just the bilinear, nondegenerate and antisymmetric mapping 
\[\sigma(\varphi_0 \oplus \pi_0,\varphi_1 \oplus \pi_1) = \int{d^3x (\varphi_0 \pi_1 - \pi_0 \varphi_1)}\]
for every $\varphi_0 \oplus \pi_0, \varphi_1 \oplus \pi_1 \in \mathcal{Q}$, which is just equation (\ref{eq:symp}). Now, we observe the following definition.
\begin{definition}
A diffeomorphism $T: \mathcal{Q}\rightarrow \mathcal{Q}$ is called a \emph{symplectomorphism} if 
\[\sigma(Tf,Tg) = \sigma (f,g),\]
i.e. if T preserves the symplectic form on $\mathcal{Q}$.
\end{definition}
Denote the set of smooth real-valued functions on $\mathcal{Q}$ by $C^{\infty}(\mathcal{Q})$. In particular, every $H\in C^{\infty}(\mathcal{Q})$ gives rise to a \emph{Hamiltonian vector field} $X_H$ which is determined by
\begin{equation}
dH = \sigma(X_H, \cdot),
\end{equation}
where $d$ denotes the exterior derivative.\footnote{Note that since $T^*\mathscr{S}$ carries a natural vector space structure, the tangent space of $T^*\mathscr{S}$ at a point $f$ can be identified with $T^*\mathscr{S}$.} The dynamics of the classical system are then determined by the flow of $X_H$, which gives rise to a unique one-parameter group of symplectomorphisms $D_t: \mathcal{Q}\rightarrow \mathcal{Q}$, mapping the phase space configuration at a time $s$ to a later time $t>s$. In other words: the dynamically possible trajectories implemented by $D_t$ are just the integral curves of the Hamiltonian vector field\footnote{In fact, the Hamiltonian vector field $X_H$ is a \emph{continuous symmetry} of the system in the sense that H and $\sigma$ are preserved along the flow lines of $X_H$. In other words, $X_H$ is a timelike Killing vector field on our phase space.}  $X_H$.\\
We are now equipped with a triple $(\mathcal{Q}, \sigma, D_t)$ which entirely determines our classical theory of a massive scalar field. Let me proceed to define the following:
\begin{definition}
An \emph{observable} on the phase space $\mathcal{Q}$ is an $F \in C^{\infty}(\mathcal{Q})$, mapping a phase space configuration to the real numbers. The nondegenerate symplectic form on $\mathcal{M}$ naturally gives rise to an algebraic structure $\{\cdot,\cdot\}$ on $C^{\infty}(\mathcal{Q})$, which is just the \emph{Poisson bracket}\footnote{cf. Wald in \cite{Wald1994} and Schmid in \cite{Schmid}.}
\begin{equation}
\{F,G\}= \sigma(X_F,X_G).
\end{equation}
\end{definition}
Particularly, for every $f \in \mathcal{Q}$, the function $\sigma(f,\cdot)$ is such an observable, which obeys the Poisson bracket
\begin{equation}
\{\sigma(f,\cdot),\sigma(g,\cdot)\}= -\sigma(f,g).
\label{eq:poiss}
\end{equation}
It can be shown that the set $\{\sigma(f,\cdot): f\in \mathcal{Q}\}$ spans the whole of $C^{\infty}(\mathcal{Q})$, i.e. every classical observable on $\mathcal{Q}$ can be constructed in this way.\\
\begin{remark}
In the context of na\"{i}ve localization, the notion of localization is contingent on the relevant functions having compact support. This is because we would like the fields corresponding to localized KG quanta to be symplectically orthogonal (as we shall see later); this can be achieved only if we consider $C_c^{\infty}(\mathbb{R}^3)$, the space of smooth functions with compact support on $\mathbb{R}^3$, and replace $\mathcal{Q}$ by 
\begin{equation}
\mathcal{M} := C_c^{\infty}(\mathbb{R}^3)\oplus C_c^{\infty}(\mathbb{R}^3).
\end{equation}
Clearly, $C_c^{\infty}(\mathbb{R}^3)$ is not complete and therefore is \emph{not} a Banach manifold. However, it can be shown that $C_c^{\infty}(\mathbb{R}^3)$ is dense in the Schwartz space $\mathscr{S}$ as defined above. Thus, any rapidly decreasing function can be approximated by a smooth function with compact support to an arbitrary degree of accuracy. For the remainder of this Section, I will work with $\mathcal{M}$ and pretend that it is a symplectic structure with $\sigma$ as defined above.\footnote{One might argue that the flaws of standard localization come in at this very point in its construction. By considering the space of rapidly decreasing functions instead of $C_c^{\infty}(\mathbb{R}^3)$, we would not only ensure the mathematical well-definedness of the symplectic structure, but also acquire an intuitive notion of `approximate localization'  - at the expense of a particle ontology (see Section \ref{approxlocal}).}
\end{remark}
\subsection{First quantization}
To accomplish the transition to quantum theory, we would intuitively regard the $\sigma(f,\cdot)$ as self-adjoint operators acting on some Hilbert space, where, according to the canonical procedure of replacing Poisson brackets by the commutator, equation (\ref{eq:poiss}) would give us the commutator, and thereby an algebraic relation on the set of operators. However, since the operator $\hat{\sigma}(f,\cdot)$ need not be bounded, we cannot be certain that the set of these operators would be closed under binary composition, which entails that their commutator need not be well-defined. Following Wald in \cite{Wald1994}, we can instead introduce the bounded operator
\begin{equation}
W(f) = \text{exp}\{i\sigma(f,\cdot)\}
\end{equation}
acting on some Hilbert space $\mathcal{H}$. We then have the following.\footnote{Following Bratelli and Robinson (Theorem 5.2.8) in \cite{BR1997}.}
\begin{proposition}\emph{(Bratteli/Robinson).}\\ \label{weylprop}
For every real-symplectic vector space $(\mathcal{M}, \sigma)$ there exists a C*-algebra $\mathcal{W[\mathcal{M},\sigma}]$ generated by unitary elements $W(f), \quad f\in \mathcal{M}$ such that 
\begin{enumerate}
\item $W^*(f)=  W(-f)$
\item $W(0)=I $
\item $W(f)W(g) = e^{i\sigma(f,g)} W(f+g)$
\end{enumerate} 
where (3) is called the \emph{Weyl form} of the canonical commutation relations (CCRs). Furthermore, $\mathcal{W}[\mathcal{M},\sigma]$ is unique up to *-isomorphism.
\end{proposition}
\begin{remark}
Note that, in particular, $\mathcal{W}[\mathcal{M},\sigma] \subseteq \mathfrak{B}(\mathcal{H})$ for some $\mathcal{H}$.
\end{remark}
We have generalised the notion of an observable to a quantum mechanical setting by making use of the Weyl form of the CCRs, but we do not yet have a Hilbert space that it acts on. So in essence, we are only two steps away from completing our task of quantizing the classical system $(\mathcal{M},\sigma,D_t)$: we need to construct a Hilbert space from $\mathcal{M}$ and a one-parameter group $U_t$ of unitary time evolution operators on $\mathcal{H}$. In other words, we are looking for a so-called quantum one-particle system, which in \cite{Halvorson2001}, Halvorson defines as follows.
\begin{definition}
Let $(\mathcal{M},\sigma,D_t)$ denote a classical system. Let $K: \mathcal{M} \hookrightarrow \mathcal{H}$ denote an injective real-linear mapping of $\mathcal{M}$ into some Hilbert space $\mathcal{H}$ such that
\begin{enumerate}
	\item The complex-linear span of $K(\mathcal{M})$ is dense in $\mathcal{H}$;
	\item K preserves the symplectic form in the sense that $2\text{Im}(Kf, Kg)_{\mathcal{H}} = \sigma (f,g)$, where $(\cdot, \cdot)_{\mathcal{H}}$ is the inner product on $\mathcal{H}$;
	\item $D_t$ gives rise to a unitary one-parameter group $U_t$ of time-evolution on $\mathcal{H}$ such that $U_t K = K D_t$.
\end{enumerate}
Then, the triple $(K, \mathcal{H}, U_t)$ is called a \emph{quantum one-particle system} over $(\mathcal{M},\sigma,D_t)$. 
\end{definition}
To put it differently, we would like to map the classical phase space $\mathcal{M}$ injectively to some Hilbert space - but of course, this injective mapping should preserve the Weyl form of the CCRs. We can accomplish this by using the following theorem,\footnote{For a proof see Bratelli \& Robinson \cite{BR1997} p.20.} which builds on proposition \ref{weylprop}.
\begin{theorem}\emph{(Bogoliubov)}\\ 
For every symplectomorphism T on a symplectic real vector space $(\mathcal{M},\sigma)$ there exists a unique automorphism $\alpha_T$ on $\mathcal{W}[\mathcal{M},\sigma]$ such that, for every $W(f) \in \mathcal{W}[\mathcal{M},\sigma]$,
\begin{equation}
	\alpha_T(W(f)) = W(Tf).
\end{equation}
Automorphisms of this form are called \emph{Bogoliubov transformations}.
\label{bog}
\end{theorem}
\begin{remark}
The automorphisms of theorem \ref{bog} are just faithful \emph{representations} $(\alpha_T, \mathcal{H_{\alpha_T}})$ of the Weyl algebra $\mathcal{W}[\mathcal{M},\sigma]$, where $\mathcal{H_{\alpha_T}}$ is a suitable Hilbert space (which we still need to construct!).
\end{remark} \label{poshamiltonian}
In other words: If we find a suitable symplectomorphism $T$ on $\mathcal{M}$, this will yield a C*-algebra $\alpha_t(\mathcal{W})\cong \mathfrak{B}(\mathcal{H}_{\alpha_T})$ which also obeys the CCRs in the Weyl form. Thus, $(\alpha_T, \mathcal{H_{\alpha_T}})$ can be regarded a GNS representation of $\mathcal{W}[\mathcal{M},\sigma]$ unto itself.\\
Following Bratteli \& Robinson,\footnote{cf. \cite{BR1987} pp. 32.} the self-adjoint closure $\overline{E}$ of the Klein-Gordon operator $E = (-\nabla^2 +m^2)$ is a positive operator, i.e. its spectrum $[m^2, \infty)$ is a subset of the positive half-line. Therefore, there exists a unique self-adjoint H such that $\overline{E} = H^2$ with a spectrum $[m, \infty)$, given by
\begin{equation}
H=\sqrt{-\nabla^2 +m^2}.
\label{eq:kgop}
\end{equation}
Following Clifton and Halvorson in \cite{CliftonHalvorson2001}, we can now define the following.
\begin{definition}
A mapping J is called a \emph{complex structure} for $(\mathcal{S,\sigma})$ if
\begin{enumerate}
	\item J is a symplectomorphism;
	\item $J^2 = -I$;
	\item $\sigma(f, Jf)>0, \qquad 0\neq f \in \mathcal{M}$.
\end{enumerate}
\end{definition}
We check that, for $\varphi, \pi \in C_c^{\infty}(\mathbb{R}^3)$, the operator
\begin{equation}
J(\varphi\oplus\pi) = -H^{-1}\pi \oplus H\varphi
\label{compl}
\end{equation}
is such a complex structure on $(\mathcal{M}, \sigma)$. With this complex structure, we can turn $\mathcal{M}$ into a complex vector space $\mathcal{M}_\mathbb{C}$ by defining multiplication with a complex scalar. For every $f\in \mathcal{M}$,
\[(\alpha + i\beta)f = \alpha f + \beta J(f) \in \mathcal{M}_{\mathbb{C}}.\]
In order to arrive at a Hilbert space, all that is left to do is to complete $\mathcal{M}_{\mathbb{C}}$ with respect to a suitable complex inner product. Let us define
\begin{equation}
(f,g)_J :=\sigma(f,Jg) + i\sigma(f,g).
\label{eq:inproduct}
\end{equation}
Let $\mathcal{M}_J$ denote the Hilbert space resulting from the completion of $\mathcal{M}_\mathbb{C}$ with respect to (\ref{eq:inproduct}). Now: if $D_t$ commutes with $J$ as symplectomorphisms of $\mathcal{M}$, then $D_t$ will extend to a complex-linear operator on $\mathcal{M}_J$ in accordance with condition (3) in the definiton of quantum one-particle systems given above. Then, since $D_t$ is a symplectomorphism, we can deduce that
\[[J,D_t] = 0 \qquad \Rightarrow \qquad (D_t f, D_t g)_J = (f,g)_J \qquad \forall f,g \in \mathcal{M}_J \]
In other words: if $D_t$ commutes with J, then it follows that $D_t$ is unitary. It can be shown (although this is non-trivial) that the one-parameter group of unitary time evolution is given by $D_t= e^{itH}$. As we have seen above, the operator $H$ given in equation (\ref{eq:kgop}) is positive. By Stone's theorem, this implies that the group $D_t$ has \emph{positive energy}. Furthermore, we can make the following general definition:
\begin{definition}
The pair $(\mathcal{H},U_t)$ is called a \emph{quantum one-particle system} just in case $\mathcal{H}$ is a Hilbert space and $U_t$ is a one-parameter unitary group on $\mathcal{H}$ with positive energy.
\end{definition}
Furthermore, one can show that J is unique:\footnote{For a proof see \cite{Kay1970}.}
\begin{proposition} $ $\\
Let $D_t$ be a one-parameter group of symplectomorphisms of $(\mathcal{M},\sigma)$. If there is a complex structure $J$ on $(S,\sigma)$ such that $(\mathcal{M}_J,D_t)$ is a quantum one-particle system, then $J$ is unique.
\end{proposition}
In summary, we have quantized the space $\mathcal{M}$ of classical solutions to the KG equation by introducing a suitable complex structure, which I then used to extend the time-evolution group uniquely to a unitary time-evolution group on $\mathcal{M}_J$. Physically, our choice of $D_t$ specifies the way we decompose the space of complex solutions to the KG equation into positive and negative frequency parts relative to the ``motion'' of the system. This uniquely fixes the complex structure $J$ on $\mathcal{M}_J$, which then exclusively represents the \emph{positive} energy solutions.
\subsection{Second quantization} \label{2nd}
Now that we have a one-particle system $(\mathcal{M}_J,D_t)$ over $(\mathcal{M}, \sigma)$ we can set out to construct a unique representation of the Weyl algebra $\mathcal{W}[\mathcal{M},\sigma]$. Define the Hilbert space
\begin{equation}
\mathcal{F}(\mathcal{M}_J) = \mathbb{C}\oplus \mathcal{M}_J \oplus \mathcal{M}_J^2 \oplus \mathcal{M}_J^3...,
\label{eq:fock}
\end{equation}
where 
\begin{equation}
	\mathcal{M}_J^n = \underbrace{\mathcal{M}_J \otimes ... \otimes \mathcal{M}_J}_{\text{n times}}
\end{equation}
denotes the n-fold symmetric tensor product of $\mathcal{M}_J$. $\mathcal{F}(\mathcal{M}_J)$ is the \emph{bosonic Fock space} over $\mathcal{M}_J$ with a unique translation-invariant state
\begin{equation}
	\Omega = 1 \oplus 0 \oplus 0 \oplus ... \quad \in \mathcal{F}(\mathcal{M}_J)
\end{equation} 
called the \emph{vacuum}. Following Bratteli \& Robinson \cite{BR1997}, we can define the \emph{unbounded} creation and annihilation operators on $\mathcal{F}(\mathcal{M}_J)$ for any $f\in \mathcal{M}$ by 
\begin{equation}
a^*(f) := \bigoplus_{k=0}\sqrt{k+1}a^*_k(f)
\label{eq:create}
\end{equation}
\begin{equation}
a(f) := \bigoplus_{k=0}\sqrt{k}a_k(f)
\label{eq:annihilate}
\end{equation}
where $a_0(f) = 0$ and the actions of the mappings
\begin{equation}
	a^*_n(f): \mathcal{M}^{n-1}_J \rightarrow \mathcal{M}^n_J, \qquad a_n(f): \mathcal{M}^{n}_J \rightarrow \mathcal{M}^{n-1}_J
\end{equation}
on product vectors in $\mathcal{M}_J^n$ are just given by
\begin{equation}
	a^*_n(f)(f_1\otimes...\otimes f_{n-1}) = f\otimes f_1 \otimes... \otimes f_{n-1}
\end{equation}
\begin{equation}
	a_n(f)(f_1 \otimes... \otimes f_n) = (f,f_1)_J f_2\otimes ... \otimes f_n.
\end{equation}
We define the self-adjoint operator
\begin{equation}
\Phi(f) := \frac{1}{\sqrt{2}}(a^*(f)+a(f)), \qquad f\in \mathcal{M}
\label{eq:fieldop}
\end{equation}
which looks like a ``field operator'' in the heuristic approach.\footnote{However, as opposed to this self-adjoint operator (equation \ref{eq:fieldop}), the ``operator-valued solutions'' $\Phi(x)$ of standard QFT are mathematically not well-defined. Instead, $\Phi(x)$ is regarded as a sesquilinear form on a dense domain $\mathcal{D} \subset \mathcal{H}$ in the sense that $\langle \Psi_2, \Phi(x) \Psi_1 \rangle$ is linear in $\Psi_1$ and anti-linear in $\Psi_2$. To obtain an operator, one has to `smear out' $\Phi(x)$ with an $f\in \mathscr{S}$, a `test function'. This yields the unbounded operator $\Phi(f) = \int{\Phi(x)f(x) d^3x}$ which is defined on $\mathcal{D}$; cf. Wightman axiom \textbf{B} in \cite{Haag}.} Then, Bratteli and Robinson \cite{BR1997} show the following (notation adapted).\footnote{For a proof see \cite{BR1997} pp. 13.}
\begin{proposition}$ $\\
For each $f\in \mathcal{M}_J$, let $\Phi(f)$ denote the self-adjoint operator as defined in equation (\ref{eq:fieldop}). Moreover, let $W(f)$ denote the unitary operator obtained from exponentiation of $\Phi(f)$:
\begin{equation}
	W(f) = \text{exp}\{i\Phi(f)\}.
\end{equation}
It follows that
\begin{enumerate}
	\item $W(f)W(g) = e^{-\text{Im}(f,g)_J}W(f+g)$
	\item $\pi_J(\mathcal{W}[\mathcal{M},\sigma]):=\{W(f):f \in \mathcal{M}_J\}$ leaves no subspace of $\mathcal{F}(\mathcal{M}_J)$ invariant.
\end{enumerate}
\label{brat2}
\end{proposition}
So we have found an irreducible representation\footnote{The representation $(\pi_J, \mathcal{F}(\mathcal{M}_J))$ is also \emph{regular} in the sense that the unitary groups $t\mapsto \pi_J(W(tf))$ are strongly continuous for all $t\in \mathbb{R}$ and $f \in \mathcal{M}_J$.} $(\pi_J, \mathcal{F}(\mathcal{M}_J))$ of the Weyl algebra $\mathcal{W}[\mathcal{M},\sigma]$ on the Fock space constructed from our quantum one-particle system! Here, $\pi_J$ is just the Bogoliubov transformation that is associated with the complex structure $J$ via theorem (\ref{bog}). The vacuum expectation value of $W(f) \in \pi_J(\mathcal{W}[\mathcal{M},\sigma])$ is just
\begin{equation}
\langle \Omega, W(f)\Omega \rangle = e^{-(f,f)_J/4}, \qquad f\in \mathcal{M}.
\label{eq:vacexp}
\end{equation}
\subsection{The `standard' localization scheme} \label{standardloc}
Now we have everything we need to construct a `standard' localization scheme for quanta of the Klein-Gordon field.\footnote{Here, I will be largely following Halvorson in \cite{Halvorson2001}, preamble of Section 3.} Consider the assignment $G\rightarrow \mathcal{M}(G)$ of a spatial region $G \subseteq \mathbb{R}^3$ to the subset $\mathcal{M}(G) \subseteq \mathcal{M}_J$ of Cauchy data localized in $G$. Let $C_c^{\infty}(G)$ denote the subspace of smooth functions with support in $G$, then 
\begin{equation}
	\mathcal{M}(G)= C_c^{\infty}(G)\oplus C_c^{\infty}(G)
\end{equation}
is a real-linear subspace of $\mathcal{M}_J$. We say that the Weyl operator $W(f)$ acting on $\mathcal{F}(\mathcal{M}_J)$ is \emph{classically localized} in $G$ just in case $f\in \mathcal{M}(G)$.
\begin{remark}
The subalgebra $\mathcal{W}(G)= \pi(\mathcal{W}[\mathcal{M}(G),\sigma])'' \subseteq \pi(\mathcal{W}[\mathcal{M},\sigma])$ of operators classically localized in $G$ is the von Neumann algebra generated by Weyl operators classically localized in G. In other words: $\mathcal{W}(G)$ consists of linear combinations and limits\footnote{In the weak operator topology; see Section \ref{mathprelims}} of Weyl operators classically localized in $G$.
\end{remark}
With this von Neumann subalgebra in hand, we can make the following definition:
\begin{definition}
The \emph{standard localization scheme} is a net of local Weyl algebras $G\rightarrow \mathcal{W}(G)$, where $\mathcal{W}(G)$ is classically localized in $G\subseteq \mathbb{R}^3$.
\end{definition}
We can now check if this localization scheme obeys the criteria defined above.
\begin{itemize}
	\item In this scheme, our notion of a Weyl algebra $\mathcal{W}(G)$ localized in $G$ is directly deduced from the subspace of Cauchy data $\mathcal{M}(G)$ localized in $G$. Therefore, the net $G\rightarrow \mathcal{W}(G)$ satisfies \emph{isotony}.
	\item A similar argument holds for \emph{translation-rotation covariance}.
	\item If, for any $G_1, G_2 \subseteq \mathbb{R}^3$, $G_1 \cap G_2 = \emptyset$, then $S(G_1)$ will be symplectically orthogonal to $S(G_2)$. That is, the spatial integral in equation (\ref{eq:symp}) vanishes since each integrand contains functions of disjoint compact support. Therefore, the algebras $\mathcal{W}(G_1)$ and $\mathcal{W}(G_2)$ commute. It follows that the net $G\rightarrow \mathcal{W}(G)$ satisfies \emph{weak microcausality}.\footnote{Since the standard localization scheme only pertains to spatial regions, we need to consider weak microcausality in this context.}
\end{itemize}
So the standard localization scheme appears to satisfy our minimal set of requirements! But the good news ends there: As we will see in the next Section, this localization scheme is severely flawed.\\
\subsection{Why na\"{i}ve localization runs into trouble}
Let us get straight to the problem: The operator $H$ that was used in our definition of the complex structure J is \emph{anti-local} in the following sense:
\begin{definition}
An operator $A$ on some Hilbert space $\mathcal{H}(\mathbb{R}^3)$ of smooth functions on $\mathcal{R}^3$ is said to be \emph{anti-local} if, for any $f\in \mathcal{H}(\mathbb{R}^3)$ and for any open $G\in \mathbb{R}^3$, the two conditions 
\begin{enumerate}
	\item $\text{supp}(f)\cap G = \emptyset$
	\item $\text{supp}(Af)\cap G = \emptyset$
\end{enumerate}
are simultaneously satisfied only if $f \equiv 0$.
\end{definition}
In other words: an operator A is anti-local if it maps functions with compact support to functions with infinite tails. In \cite{SegalGoodman}, Segal and Goodman show the following:
\begin{lemma} $ $\\
The operator $H = \sqrt{-\nabla^2 + m^2}$ is anti-local.
\end{lemma}
Segal and Goodman then go on to show that, because H is anti-local, the \emph{complex}-linear span of $\mathcal{M}(G)$ is dense in $\mathcal{M}_J$ for any open subset $G\subseteq \mathbb{R}^3$. Now: consider the algebra $\mathfrak{R}$ generated by $\{W(f): f\in E\}$ where $E \subseteq \mathcal{M}_J$ is a real-linear subspace. But, as proven in \cite{Petz}, $\Omega$ is cyclic for $\mathfrak{R}$ if and only if the complex-linear span of $E$ is dense in $\mathcal{M}_J$. Therefore, anti-locality of $H$ entails that the vacuum is cyclic for every local algebra. 
\begin{theorem}\emph{(Reeh-Schlieder in $\mathbb{R}^3$)}\\
Let G be proper subset of $\mathbb{R}^3$. Then, the vacuum $\Omega$ is cyclic with respect to any local algebra $\mathcal{W}(G)$ of the net $G\rightarrow \mathcal{W}(G)$.
\label{rsr3}
\end{theorem}
Thus, local operations on the vacuum can approximate the entire Hilbert space $\mathcal{M}_J$! This suggests that our `standard' localization scheme does not comply with the laws of special relativity: if \emph{actions} confined to a spatial region cause \emph{effects} at a spacelike separation, then this implies superluminal, if not \emph{instantaneous} propagation of wave packets. This, as Fleming and Butterfield point out in \cite{FB1999}, ``...is certainly hard to square with na\"{i}ve, or even educated, intuitions about localization!''\\
However, this is not the only strange consequence of the Reeh-Schlieder theorem. I will now give three other strange properties of this localization scheme: The first pertaining to `standard' localization in general, the other two pertaining concretely to the localization of Klein-Gordon quanta in a spatial set.
\subsubsection{Entanglement of the vacuum}
Let $\mathcal{W}(G_1)$ and $\mathcal{W}(G_2)$ denote von Neumann algebras of operators classically localized in $G_1$ and $G_2$, respectively, where $G_1 \cap G_2 = \emptyset$ and $G_1, G_2 \subseteq \mathcal{R}^3$. Since the subspace $\mathcal{M}(G_1)$ and $\mathcal{M}(G_2)$ are symplectically orthogonal, we have $(f,g)_J = \text{Re}(f,g)_J$ for $f\in \mathcal{M}(G_1), \quad g\in \mathcal{M}(G_2)$. According to equation (\ref{eq:vacexp}), and using the Weyl CCRs (\ref{brat2}), we find that 
\begin{equation}
	\langle \Omega, W(f)W(g)\Omega \rangle = \langle \Omega, W(f)\Omega \rangle \langle \Omega, W(g) \Omega \rangle e^{-\text{Re}(f,g)_J}.
\end{equation}
If the vacuum were a product state across the spacelike separated regions $G_1$ and $G_2$, the real part of the inner product would be zero. However, this is generally not the case. I will make this obvious by the following: for $f= \varphi_0 \oplus \pi_0$ and $g = \varphi_1 \oplus \pi_1$, the integral given by
\begin{equation}
Re(f,g)_J = \sigma(f, Jg) = \int{d^3x\varphi_0(H\varphi_1)}+ \int{d^3x\pi_0(H^{-1}\pi_1)}
\end{equation}
does not generally vanish due to the anti-locality of H and its inverse. It follows that the vacuum is not a pure state across spacelike separated regions. In order to show that the vacuum is, in fact, \emph{entangled} across spacelike separated regions, let me make the following definition:\footnote{Following Clifton and Halvorson in \cite{CliftonHalvorson2000}.}
\begin{definition}
Let $O_1, O_2\subseteq M$ be spacelike separated double cones in Minkowski spacetime and let $\{O\rightarrow \mathfrak{R}(O): \mathfrak{R}(O) \subseteq \mathfrak{B}(H)\}$ denote a net of von Neumann algebras. A state $\omega$ is said to be \emph{entangled} across $O_1, O_2$ just in case the restriction of $\omega$ to the algebra
\begin{equation}
\mathfrak{R}_{O_1 O_2} := \{\mathfrak{R}(O_1)\cap \mathfrak{R}(O_2)\}''
\end{equation} falls outside the weak *-closure\footnote{The \emph{weak *-topology} on the state space of a von Neumann algebra $\mathfrak{R}$ is defined as follows: a sequence or net of states $\{\omega_n\}$  converges to a state $\omega$ just in case $\omega_n(Z) \rightarrow \omega(Z), \quad \forall Z\in \mathfrak{R}$.} of the convex hull of pure states on $\mathfrak{R}_{O_1 O_2}$.
\end{definition}
\begin{proposition} $ $\\
Let $\Omega \in \mathcal{F}(\mathcal{M}_J)$ be the vacuum state of the Fock space. The abstract state $\omega_{\Omega}$ on the algebra $\pi(\mathcal{W}[\mathcal{M}, \sigma])$ constructed via theorem (\ref{gns}) 
\begin{equation}
\omega_{\Omega} := \langle \Omega, \pi(W(f)) \Omega \rangle \quad \forall W(f) \in \pi(\mathcal{W}[\mathcal{M}, \sigma])
\end{equation}
is \emph{entangled} across $\mathcal{W}(G_1)$ and $\mathcal{W}(G_2)$, where $G_1, G_2 \in \mathbb{R}^3$ are spacelike.
\end{proposition}
Before I give a proof of this proposition, we need to observe the following
\begin{remark}
Let $G_1, G_2 \in \mathbb{R}^3 $ be spacelike sets and $\mathcal{W}_{G_1 G_2} := \{\mathcal{W}(G_1)\cap \mathcal{W}(G_2)\}$. Then, the state space on the algebra $\mathcal{W}_{G_1 G_2}$ does not contain any entangled states if and only if $\mathcal{W}(G_1)$ and $\mathcal{W}(G_2)$ are \emph{abelian}.
\end{remark}
\begin{proof}
Let $\mathcal{W}(G_1)$ and $\mathcal{W}(G_2)$ as above where $G_1, G_2 \subseteq \mathbb{R}^3$ are spacelike, and assume that $\omega_{\Omega}$ is \emph{not} entangled across $\mathcal{W}(G_1)$ and $\mathcal{W}(G_2)$. Then, we can make the following argument:
\begin{enumerate}
	\item Local operations represented by the algebra $\mathcal{W}(G)$ cannot turn $\omega_{\Omega}$ into a state entangled across $\mathcal{W}(G_1)$ and  $\mathcal{W}(G_2)$.
	\item Cyclicity of $\Omega$ entails that local operations on $\omega_{\Omega}$ can approximate the entire state space on $\mathcal{W}_{G_1 G_2}$;
	\item Therefore, no state on $\mathcal{W}_{G_1 G_2}$ is entangled across $\mathcal{W}(G_1)$ and $\mathcal{W}(G_2)$;
	\item Therefore, $\mathcal{W}(G_1)$ and $\mathcal{W}(G_2)$ must be abelian.
\end{enumerate}
in contradiction to the fact that the (local) Weyl algebras $\mathcal{W}(G_1)$ and $\mathcal{W}(G_2)$ are non-abelian by construction. Therefore, $\omega_{\Omega}$ is entangled across $\mathcal{W}(G_1)$ and $\mathcal{W}(G_2)$ and it follows that the vacuum $\Omega$ entangled across spacelike $G_1, G_2 \in \mathbb{R}^3$. $\blacksquare$
\end{proof}
Additionally, we observe that since the net $\{\mathcal{W}(G): G\subseteq\mathbb{R}^3\}$ satisfies weak microcausality, the cyclicity of $\Omega$ entails that the vacuum is also \emph{separating} for any local algebra $\mathcal{W}(G)$, where $G' \neq \emptyset$. Therefore, if we let $A\in \mathcal{W}(G)$ be a nonzero projector representing the probability that an event will occur, the probability that this event will occur in the vacuum state is nonzero, since $A\Omega = 0$ would entail $A=0$, contradicting the fact that $A$ is nonzero.
\subsubsection{No local creation and annihilation operators}
Recall the definition given in Section (\ref{fundcrit}): An unbounded operator A is affiliated with a local algebra $\mathfrak{R}$ just in case unitary operations contained in the commutant $\mathfrak{R}'$ leave A unchanged. Since the annihilation and creation operators specified in (\ref{eq:create}) and (\ref{eq:annihilate}) are unbounded, we need to observe the following:
\begin{remark}
By construction, the operator $\Phi(f)$ is affiliated with $\mathcal{W}(G)$ just in case $W(f) \in \mathcal{W}(G)$.
\end{remark}
\begin{remark}
By definition, the family of operators affiliated with a local algebra $\mathcal{W}(G)$ is closed under taking adjoints.
\end{remark}
We know that annihilation operators $a(f)$ annihilate the vacuum, $a(f)\Omega = 0$. Therefore, (in the sense defined above) the operators $a(f)$ are \emph{not} affiliated with any local algebra $\mathcal{W}(G)$. Since the family of operators affiliated with $\mathcal{W}(G)$ is closed under taking adjoints, there are also no annihilation operators $a^*(f)$ affiliated with $\mathcal{W}(G)$.\\
Furthermore: using the anti-linearity of $a(f)$ and the linearity of $a^*(f)$, one can invert expression (\ref{eq:fieldop}):
\begin{equation}
a^*(f) = \frac{1}{\sqrt{2}}(\Phi(f) - i\Phi(if)), \qquad a(f) = \frac{1}{\sqrt{2}}(\Phi(f) + i\Phi(if)).
\label{eq:invfieldop}
\end{equation}
From this, we can see why problems arise from standard localization: the subspace of Cauchy data $\mathcal{M}(G)$ localized in a spatial region $G \subseteq \mathbb{R}^3$ is a \emph{real}-linear subspace of $\mathcal{M}_J$. However, from equation (\ref{eq:invfieldop}), we can see that $a^*(f)$ and $a(f)$ will be affiliated with a local Weyl algebra $\mathcal{W}(G)$ only if $f\in E$, where $E\subseteq \mathcal{M}_J$ is a \emph{complex}-linear subspace of the Hilbert space $\mathcal{M}_J$.
\subsubsection{No local number operators}
Similarly, there are no number operators affiliated with local algebras $\mathcal{W}(G)$. Two arguments can be given for this: First, we observe that the local number operators, defined by 
\begin{equation} 
N(f) = a^*(f)a(f), \quad \forall f \in \mathcal{M}(G), \quad G\in \mathbb{R}^3,
\end{equation}
also annihilate the vacuum. Furthermore, due to the anti-linearity of $a(f)$ and the linearity of $a^*(f)$, the number operator is invariant under phase transformations of $f$:
\begin{equation}
N(f) = N(e^{it}f), \quad \forall t \in \mathbb{R}.
\end{equation}
But: from the way we constructed classical localization of any $f\in \mathcal{M}$, it is obvious that the \emph{localization} of any f is not invariant under such a phase transformation. Therefore, it is not even possible to give a well-defined local number operator! Therefore, we can state the following.
\begin{result}
The standard localization scheme fails to satisfy the \emph{fundamentality} criterion for the existence of a particle ontology. 
\label{standard}
\end{result}

\section{The Newton-Wigner approach}\label{NWloc}
This of course, is an old hat: As early as 1949, Newton and Wigner\footnote{See \cite{NW1949}.} came up with a localization scheme that - at a first sight - appears to alleviate the problems of standard localization. In this Section,\footnote{cf. Halvorson in \cite{Halvorson2001}, pp. 14-16.} I will first report this localization scheme in the algebraic context and explain how it constitutes an improvement to the status of particles. However, I will then go on to show that the NW scheme is still subject to a generalized version of the Reeh-Schlieder theorem and thereby also fails to satisfy the \emph{fundamentality} criterion.
\subsection{First and second quantization, NW style}
In the previous Section, I constructed a representation of the Weyl algebra $\mathcal{W}[\mathcal{M}, \sigma]$ by endowing $\mathcal{M}$ with a complex structure $J$ and completing it with respect to the complex norm $(\cdot, \cdot)_J$. Thereby, we obtained a quantum one-particle system $(\mathcal{M}_J, D_t)$. In this subsection, I will construct another quantum one-particle system over $(\mathcal{M},\sigma)$ and thereby provide the \emph{Newton-Wigner} representation of $\mathcal{W}[\mathcal{M}, \sigma]$.\\
Given our symplectic structure $(\mathcal{M},\sigma, D_t)$, let us introduce an injective mapping $K: \mathcal{M} \hookrightarrow L^2(\mathbb{R}^3)$ such that, for every $f:= \varphi\oplus \pi \in \mathcal{M}$:
\begin{equation}
K(\varphi\oplus \pi)= \frac{1}{\sqrt{2}}(H^{1/2}\varphi + iH^{-1/2}\pi).
\end{equation}
Clearly, the complex-linear span of $K(\mathcal{M})$ is dense in $L^2(\mathbb{R}^3)$. With regard to the symplectic form, we observe the straightforward calculation: For all $f:=\varphi_0 \oplus \pi_0, g:= \varphi_1\oplus \pi_1 \in \mathcal{M}$,
\begin{eqnarray}
	2\text{Im}(Kf,Kg)_{L^2}& = &2\text{Im}\left \{  \int_{\mathbb{R}^3}   Kf(Kg)^* \mathrm d^3x \right \} \\
																	 & = & \int_{\mathbb{R}^3} \left(\varphi_0 \pi_1 - \varphi_1 \pi_0 \right ) \mathrm d^3x \\
																	 & = & \sigma (f,g).
\end{eqnarray}
Furthermore, it can be shown rigorously that $K$ gives rise to a one-parameter group of unitary time evolution on $L^2(\mathbb{R}^3)$. Here, it should suffice to motivate this with the following: Let $f,g \in \mathcal{M}$. Because $D_t$ is a symplectomorphism, we have
\begin{eqnarray}
\nonumber 
	\sigma (f, g) = \sigma (D_t f, D_t g)& = & 2i \left \{ \text{Re}(K D_t f, K D_t g)_{L^2} - (K D_t f, K D_t g)_{L^2} \right \} \\  
	              			                 & = & 2i \left \{ \text{Re}(K f, K g)_{L^2} - (K f, K g)_{L^2} \right \} \\ \nonumber 
											                 & = & 2i \left \{ \text{Re}(U_t K f, U_t  K g)_{L^2} - (U_t K f,U_t K g)_{L^2} \right \}.
\end{eqnarray}
That is, $KD_t = U_tK$ for a suitable unitary transformation $U_t$ on $L^2(\mathbb{R}^3)$, where the generator $H$ of $U_t$ is a positive operator (recall Section \ref{poshamiltonian}). Therefore, we arrive at a quantum one-particle system $(K, L^2(\mathbb{R}^3), U_t)$.
\begin{remark}
The quantum one-particle systems $(K, L^2(\mathbb{R}^3), U_t)$ and $(J, \mathcal{M}_J, D_t)$ are unitarily equivalent. However, unlike before where $\mathcal{M} \subseteq \mathcal{M}_J$,  $\mathcal{M} \nsubseteq L^2(\mathbb{R}^3)$. 
\end{remark}
The real-linear invertible mapping $K$ gives rise to a \emph{Bogoliubov}-Transformation on $\mathcal{W}[\mathcal{M},\sigma]$ via theorem (\ref{bog}). Thereby, we obtain the representation 
\begin{equation}
\alpha_K \left (\mathcal{W}[\mathcal{M},\sigma] \right )  = \mathfrak{B}(L^2(\mathbb{R}^3))
\label{eq:NW}
\end{equation}
of $\mathcal{W}[\mathcal{M},\sigma]$. The second quantization proceeds in analogy to Section (\ref{2nd}). The Fock space is given by
\begin{equation}
\mathcal{F}[L^2(\mathbb{R}^3)] = \mathbb{C}\oplus \mathcal{H} \oplus \mathcal{H}^2 \oplus \mathcal{H}^3...,
\label{eq:fockl2}
\end{equation}
where $\mathcal{H} := L^2(\mathbb{R}^3)$. Again, the vacuum is given by the unique translation-invariant state
\begin{equation}
\Omega = 1 \oplus 0 \oplus 0 \oplus ... \quad \in \mathcal{F}[L^2(\mathbb{R}^3)]
\end{equation}
and we can define creation and annihilation operators $a^*(f), a(f)$ as above, where $f\in L^2(\mathbb{R}^3)$. By proposition (\ref{brat2}) we have a representation $(\alpha_K, \mathcal{F}[L^2(\mathbb{R}^3)])$ of the Weyl algebra $\mathcal{W}[\mathcal{M},\sigma]$. The vacuum expectation value of any $W(f) \in \alpha_K \left (\mathcal{W}[\mathcal{M},\sigma]) \right )$ is just
\begin{equation}
\langle \Omega, W(f) \Omega \rangle = e^{-(f,f)_{L^2}/4}.
\label{eq:vacexpNW}
\end{equation}
\subsection{The NW localization scheme}
Again, we can use this representation to construct a localization scheme. In analogy to Section (\ref{standardloc}), consider the assignment $G \rightarrow L^2(G)$ of a spatial region $G \subseteq \mathbb{R}^3$ to the subset $L^2(G) \subseteq L^2(\mathbb{R}^3)$ of square-integrable functions with probability amplitude vanishing (almost everywhere) outside $G$. We say that an operator $W(f)$ acting on $\mathcal{F}[L^2(\mathbb{R}^3)]$ is \emph{NW-localized} in $G$ just in case $f \in L^2(G)$. 
\begin{remark} 
The subalgebra $\mathcal{W}_{NW}(G) = \{W(f):f \in L^2(G)\}''$ of operators NW-localized in G is the von Neumann algebra generated by Weyl operators NW-localized in G.
\end{remark}
\begin{definition}
The \emph{Newton-Wigner localization scheme} is a net of Weyl algebras $G\rightarrow \mathcal{W}_{NW}(G)$, where $\mathcal{W}_{NW}(G)$ is NW-localized in $G\subseteq \mathbb{R}^3$. 
\end{definition}
Just like the standard localization scheme, this NW localization pertains to \emph{spatial} localization. Therefore, let us first check whether this localization scheme satisfies assumptions (\ref{isotony}), (\ref{trcovariance}) and (\ref{wm}).
\begin{itemize}
	\item By construction, the NW localization scheme obeys \emph{isotony}.
	\item It is straightforward to check that the NW algebras are covariant under actions of the Euclidean group and hence satisfies \emph{translation-rotation covariance}.\footnote{In fact, Newton and Wigner constructed the ``position eigenstates'' starting from the requirement that actions of the Euclidean group would render localized states orthogonal to one another; see \cite{NW1949}.}
	\item Since for $G_1\cap G_2 =\emptyset$, $L^2(G_1)$ and $L^2(G_2)$ are orthogonal with respect to the inner product on $L^2(\mathbb{R}^3)$, this localization scheme satisfies \emph{weak microcausality}.
\end{itemize}
As we shall see, the advantages of the Newton-Wigner approach are largely due to the properties of the Hilbert space $L^2(\mathbb{R}^3)$. Let me first show how the NW localization scheme is \emph{not} subject to the Reeh-Schlieder theorem in $\mathbb{R}^3$:
\begin{proposition} $ $\\
Let $\{\mathcal{W}_{NW}(G): G\subseteq \mathbb{R}^3\}$ denote the NW localization scheme in $\mathbb{R}^3$ that satisfies assumptions (\ref{isotony}), (\ref{trcovariance}) and (\ref{wm}). Then, the vacuum $\Omega \in \mathcal{F}[L^2(\mathbb{R}^3)]$ is \emph{not} cyclic with respect to $\mathcal{W}_{NW}(G)$ for all $G \subseteq \mathbb{R}^3$.
\end{proposition}
\begin{proof}
Let $G \subseteq \mathbb{R}^3$ be an open set. Then, for $G' =\mathbb{R}^3 \setminus G$ we can decompose $L^2(\mathbb{R}^3)$ in the following way:
\begin{equation}
 L^2(\mathbb{R}^3) = L^2(G \cup G') = L^2(G) \oplus L^2(G').
 \end{equation}
Denoting $\mathcal{F}_G = \mathcal{F}[L^2(G)]$ and $\mathcal{F}_{G'}= \mathcal{F}[L^2(G')]$, it follows from the associativity of the direct sum that the Fock space over $L^2(\mathbb{R}^3)$ decomposes into
\begin{equation}
\mathcal{F}[L^2(\mathbb{R}^3)] \cong \mathcal{F}_G \otimes \mathcal{F}_{G'}.
\label{eq:prodstate}
\end{equation}
Assume that $\Omega$ is cyclic for $\mathcal{W}_{NW}(G)$, that is $\text{span}_{\mathbb{C}}\overline{\{A\Omega: A \in \mathcal{W}_{NW}(G)\}} = \mathcal{F}[L^2(\mathbb{R}^3)]$.
Then, we have the following argument:  
\begin{enumerate}
	\item Because of equation (\ref{eq:prodstate}) the vacuum vector $\Omega \in \mathcal{F}[L^2(\mathbb{R}^3)]$ is a product state across $G$ and $G'$, that is 
		\begin{equation}
			\Omega = \Omega_G \oplus \Omega_{G'},
		\end{equation} 
		where $\Omega_G \in \mathcal{F}_G $ and $\Omega_{G'} \in \mathcal{F}_{G'}$. 
	\item We can express $\mathcal{W}_{NW}(G)$ as the set of bounded operators acting exclusively on $L^2(G)$, \begin{equation}\mathcal{W}_{NW}(G) \cong \mathfrak{B}(L^2(G)) \otimes I.\end{equation}
	\item Therefore, $\text{span}_{\mathbb{C}}\overline{\{A\Omega: A \in \mathcal{W}_{NW}(G)\}} = \mathcal{F}_G \otimes \Omega_{G'}$,
\end{enumerate}
in contradiction to the assumption. Therefore, $\Omega$ is not cyclic for the Newton-Wigner localization scheme. $\blacksquare$
\end{proof}
Furthermore, the Newton-Wigner scheme comes with a number of advantages over the `standard' localization scheme:\\
\noindent\hspace*{5mm}\textbf{1.}	As we have seen, for NW localization, the vacuum is a product state across disjoint regions in $\mathbb{R}^3$. Another way to see this is by observing that NW operators local in disjoint regions of space will commute with respect to $(\cdot,\cdot)_{L^2}$, the inner product on $L^2(\mathbb{R}^3)$. Therefore, the vacuum expectation value (\ref{eq:vacexpNW}) of $W(f)W(g)$ is just
\begin{equation}
\langle \Omega, W(f)W(g) \Omega \rangle = \langle \Omega, W(f) \Omega \rangle\langle \Omega, W(g) \Omega \rangle,
\end{equation}
where $W(f) \in \mathcal{W}_{NW}(G)$ and $W(g) \in \mathcal{W}_{NW}(G')$ and $G\cap G' = \emptyset$.\\
\noindent\hspace*{5mm}\textbf{2.} Next, we can look at the annihilation and creation operators. Recall equation (\ref{eq:invfieldop}):
\begin{equation}
a^*(f) = \frac{1}{\sqrt{2}}(\Phi(f) - i\Phi(if)), \qquad a(f) = \frac{1}{\sqrt{2}}(\Phi(f) + i\Phi(if)).
\end{equation}
Clearly, $L^2(G)$ is a \emph{complex}-linear subspace of $L^2(\mathbb{R}^3)$. Therefore, for $f\in L^2(G)$, \emph{both} $\Phi(f)$ and $\Phi(if)$ are affiliated with $\mathcal{W}_{NW}(G)$. It follows that the family of creation and annihilation operators is affiliated with $\mathcal{W}_{NW}(G)$ just in case $f \in L^2(G)$. In other words: the creation and annihilation of KG-quanta is a local operation!\\
\noindent\hspace*{5mm}\textbf{3.} A similar argument holds for the number operators constructed from the above annihilation and creation operators. Since the family of annihilation and creation operators generated by $a^*(f), a(f)$ is affiliated with the NW-local algebra $\mathcal{W}_{NW}(G)$ for all $f \in L^2(G)$, the number operator
\begin{equation}
N(f) := a^*(f)a(f)
\end{equation}
is also affiliated with $\mathcal{W}_{NW}(G)$. Therefore, the \emph{total number of particles in a spatial set} is locally measurable! So at first glance, Newton Wigner appears to satisfy the \emph{fundamentality} criterion for particles in QFT. However, as we shall see in the next Section, this conclusion is premature.
\subsection{Is Newton-Wigner immune against the Reeh-Schlieder?}
First of all, it is a welcome property that the vacuum in the NW scheme is not cyclic with respect to any NW algebra local in a spatial set. However, it is not clear why it would be less unsettling for any other vector in $\mathcal{F}[L^2(\mathbb{R}^3)]$ to be cyclic for any NW-local algebra. In fact, as I will show in a moment, there is still a dense set of vectors in $\mathcal{F}[L^2(\mathbb{R}^3)]$ that are cyclic for any NW-local algebra. Subsequently, I will show that the Newton-Wigner localization scheme is subject to the Reeh-Schlieder theorem formulated in terms of double cones in Minkowski spacetime.
\subsubsection{Cyclic vectors for the spatial NW scheme}
First of all, let me state the following theorem:
\begin{theorem}\emph{(Clifton et al.)}\\
Let $\mathcal{F} \cong \mathcal{H}^0 \otimes \mathcal{H}^1 \otimes \mathcal{H}^2 \otimes ... \otimes \mathcal{H}^n$ denote a state space where $n > 1$ and each $\mathcal{H}$ is separable and has nontrivial $(>1)$ dimension. Then, the following is equivalent:
\begin{enumerate}
	\item There exists a cyclic vector in $\mathcal{F}$.
	\item All $H^k$ have the same dimension, and for $n>2$, their common dimension is infinite.
	\item The set of cyclic vectors is dense in $\mathcal{F}$.
\end{enumerate}
\end{theorem}
\begin{proof} 
(see Clifton et al. in \cite{Clifton1998}). $\blacksquare$
\end{proof}
Using this, we can show that while the vacuum is not cyclic for the spatial NW localization scheme, there is still a dense set of cyclic vectors for any NW-local algebra:
\begin{theorem} $ $\\
Let $\{\mathcal{W}_{NW}(G): G \subseteq \mathbb{R}^3\}$ denote the NW localization scheme. For any $G \subseteq \mathbb{R}^3$, $\mathcal{W}_{NW}(G)$ has a dense set of cyclic vectors in $\mathcal{F}[L^2(\mathbb{R}^3)]$.
\end{theorem}
\begin{proof}
Recalling equation (\ref{eq:prodstate}), we can write the Fock space as $\mathcal{F}[L^2(\mathbb{R}^3)] \cong \mathcal{F}(G)\otimes \mathcal{F}(G')$ for $G \cup G' = \mathbb{R}^3$. Clearly, $\mathcal{F}(G)$ and $\mathcal{F}(G')$ have infinite dimension. Therefore, the NW algebra $\mathcal{W}_{NW}(G) = \mathfrak{B}(L^2(G)) \otimes I$ local in G has a dense set of cyclic vectors in $\mathcal{F}[L^2(\mathbb{R}^3)]$. $\blacksquare$
\end{proof}
Therefore, the NW localization scheme is still highly non-local since one can still approximate the entire state space via local operations. As we shall see, a slight generalization of the Reeh-Schlieder theorem shows that NW algebras local in a region of spacetime have all the undesirable properties of the standard localization scheme. In the next Section, I will report the Reeh-Schlieder theorem in its most general form\footnote{For the next subsection, cf. Halvorson in \cite{Halvorson2001}, p. 17-21 and \cite{HalvorsonMueger2007}.} and explicate its consequences for the NW localization scheme.
\subsubsection{NW and the Reeh-Schlieder theorem in spacetime}
Recall that our construction of the quantum one-particle systems $(\mathcal{M}_J, D_t)$ and $(L^2(\mathbb{R}^3), U_t)$ essentially corresponded to selecting the subspace of positive-energy solutions to the KG equation. We accomplished this via a positive operator H that gave rise to the complex structure $J$ on $\mathcal{M}$ and an injection from $\mathcal{M}$ to $L^2(\mathbb{R}^3)$. Furthermore, the positivity of H ensured that the a one-parameter group of unitary time evolution on the respective Hilbert spaces had \emph{positive energy}.\\
In the case of a relativistic 4-dimensional spacetime, we need to be slightly more careful with our approach. In this setting, the above requirement amounts to the assumption that the energy of a state should be positive in \emph{every} Lorentz frame. If we consider the representation of the translation group on the automorphism group on our net, we can formulate this requirement as a statement about the spectrum of the generator $\mathbf{P}$ of this representation:\footnote{cf. Halvorson in \cite{HalvorsonMueger2007}.}
\begin{assumption}\emph{(Spectrum condition)}\\
Let $\mathcal{G}$ denote the translation group on Minkowski space and let $\mathfrak{A}$ be a net of local algebras. A representation $\mathbf{x} \rightarrow \alpha_{\mathbf{x}} = e^{i \mathbf{x}\cdot \mathbf{P}}$ of $\mathcal{G}$ on the automorphism group $\text{Aut}\mathfrak{A}$ is said to satisfy the \emph{spectrum condition} if its generator $\mathbf{P}$ is contained in the forward lightcone.
\label{spectrum}
\end{assumption}
This assumption represents to basic assumptions about relativistic spacetimes, namely that (1) there can be no superluminal propagation of physical effects, and (2) that energy is positive.  Thus, we are now in a position to give the full Reeh-Schlieder theorem.
\begin{theorem}\emph{(Reeh-Schlieder)}\\
Let $\mathfrak{A} = \{\mathcal{R}(O): O \subseteq M\}$ denote an additive net of local algebras that satisfies isotony and the spectrum condition. Then, for any open region, the vacuum $\Omega$ is cyclic for $\mathfrak{R}(O)$.
\label{reehschlieder}
\end{theorem}
\begin{proof}
(see Reeh and Schlieder in their \cite{ReehSchlieder} or Haag in \cite{Haag}). $\blacksquare$
\end{proof}
In the sense that the operator $H = \sqrt{-\nabla^2 +m^2}$ as defined in (\ref{eq:kgop}) is positive, we can regard theorem (\ref{rsr3}) as a special case of this more general version. By construction, the NW localization scheme pertains to spatial sets. If we want to check if the NW scheme is subject to the generalized Reeh-Schlieder theorem, we need to get an idea of what it means for something to be NW-localized in a region of \emph{spacetime}. Following \cite{Halvorson2001}, I shall adopt the ``Heisenberg picture'' where the states on a NW-local algebra $\mathcal{W}_{NW}(G)$ are fixed for all time, but the algebras evolve unitarily via 
\begin{equation}
U^{-1}(t) \mathcal{W}_{NW}(G) U(t),
\end{equation}
which yields an algebra of NW operators localized in $G$ at a time $t$. Particularly, the vacuum vector is the same for all time. Given this notion of a time-evolved NW-local algebra, we can define the algebra of operators NW-local in a region G \emph{for a finite time interval} $\Delta \subseteq \mathbb{R}$:
\begin{equation}
\mathfrak{R}_{\Delta} = \{U^{-1}(t) \mathcal{W}_{NW}(G) U(t): t \in \Delta\}.
\end{equation}
\begin{theorem} $ $\\
For any interval $\Delta = (a,b)$ around $0$, $\Omega$ is cyclic for $\mathfrak{R}_{(a,b)}$.
\end{theorem}
\begin{proof} (Sketch)
Clearly, $\Omega$ is cyclic for $\mathfrak{R}_{\mathbb{R}}$, that is: the vacuum is cyclic for NW algebras local in $G$ for all times. Kadison showed in \cite{Kadison} that since the time-translation group on $\mathcal{F}(L^2(\mathbb{R}^3))$ has positive energy, it follows that 
\begin{equation}
\text{span}_{\mathbb{C}}\overline{\{A\Omega: A \in \mathfrak{R}_{(a,b)}\}} = \text{span}_{\mathbb{C}}\overline{\{A\Omega: A \in \mathfrak{R}_{\mathbb{R}}\}}.
\end{equation} Therefore, $\Omega$ is cyclic for $\mathfrak{R}_{(a,b)}$. $\blacksquare$
\end{proof}
Therefore, the vacuum is cyclic with respect to any NW-algebra local in a spatial set $G$ over a finite time interval. In other words: operations local in a region of minkowski space can have effects at spacelike separation, since they can approximate the entire state space! I shall conclude my argument against the Newton-Wigner localization scheme by showing how it fails to satisfy the necessary condition for the existence of a particle ontology proposed in Section (\ref{fundcrit}).\\
Let $G_1$ and $G_2$ be disjoint spatial sets and define 
\begin{equation}
O_1 = \{G_1 + t: t \in (a,b) \text{ and }G \subseteq \mathbb{R}^3\},
\end{equation} 
\begin{equation}
O_2 = \{G_2 + t: t \in (a,b) \text{ and }G \subseteq \mathbb{R}^3\}
\end{equation} 
such that $O_1$ and $O_2$ are spacelike. Furthermore, define
\begin{equation}
\mathfrak{R}_{NW}(O_i) = \{U^{-1}(t) \mathcal{W}_{NW}(G) U(t): t \in (a,b)\}.
\end{equation}
\begin{proposition}
The $\mathfrak{R}_{NW}(O_i)$ as defined above to \emph{not} satisfy strong microcausality.
\end{proposition}
\begin{remark}
Consider some local algebra $\mathcal{R}(O)$ with $O' \neq \emptyset$. If a vector $\psi$ is cyclic for $\mathcal{R}(O')$, then it must be separating for $\mathcal{R}'(O')$. Strong microcausality entails that $\mathcal{R}(O) \subseteq \mathcal{R}'(O')$. Therefore, $\psi$ is also separating for $\mathcal{A}(O)$. In other words: The cyclicity of a vector for some local algebra implies that $\psi$ is separating for the algebra \emph{if and only if} the algebra satisfies strong microcausality.
\end{remark}
\begin{proof} \emph{(Prop. 4.7)}
Let $\mathfrak{R}_{NW}(O_i)$ as defined above. Recall that, since $L^2(G)$ is a complex-linear subspace of $L^2(\mathbb{R}^3)$, any spatially local NW-algebra is affiliated with a family of creation and annihilation operators. However, by construction, if $\mathcal{W}_{NW}(G_i)$ is affiliated with $a^*(f)$ and $a(f)$, then so is the algebra $\mathfrak{R}_{NW}(O_i)$ and therefore, the vacuum is \emph{not} separating for $\mathfrak{R}_{NW}(O_i)$. It then follows from the above remark that $\mathfrak{R}_{NW}(O_i)$ does not satisfy strong microcausality. $\blacksquare$
\end{proof}
If we were to interpret a NW-local algebra as an algebra of observables measurable in a spatial set $G$, we therefore would be confronted with the problem of spacelike-distant effects of local operations contained $G$ - inconsistent with special relativity. Proponents of the Newton-Wigner representation have contested assumption (\ref{sm}), which I used to give a \emph{fundamentality} criterion, viz. a necessary condition for the existence of a particle ontology. In their paper \cite{FB1999}, Fleming and Butterfield give an argument in support of the NW representation and propose a way to make physical sense of NW operators despite the counterintuitive features shown above. They proceed to introduce Newton-Wigner operators pertaining to the \emph{centre of energy} (CE) or \emph{centre of mass} (CM). These physical quantities, although \emph{localized} in a spatial region, are \emph{not} locally measurable, since any measurement process would not be confined to that region. Thus, the violation of strong microcausality by the associated NW-local algebras could be regarded as a natural consequence of the inherent non-locality of the physical quantities.\\
However, it should not be surprising that algebras of local observables pertaining to non-local physical quantities are themselves not in agreement with locality! In fact, it is not clear how to interpret the notion of a \emph{local} algebra pertaining to \emph{non-local} physical quantities, since it does not make sense to interpret any observable affiliated with a non-local physical quantity to be an element of \emph{any} local algebra! Therefore, I will uphold my definition of \emph{local measurability} given in Section (\ref{fundcrit}), which leads to the following result.
\begin{result}
The Newton-Wigner localization scheme fails to satisfy the fundamentality criterion for the existence of a particle ontology.
\end{result}
This gives us an intermediate conclusion: 
\begin{quote}
Free bosonic QFT does not permit an ontology contingent on the localization of KG-quanta or a property of a system in regions of space or spacetime.
\end{quote}

\section{Particles as an emergent phenomenon}
So far, I have shown that if we accept a set of physical assumptions as given, a free bosonic QFT cannot be interpreted as a theory about \emph{fundamental} particles. However, considering that we do experience particle localization events, and that the particle picture is used frequently in any physical phenomena in QFT, a simple rejection of particles for the theory is generally regarded to be unsatisfactory. As Halvorson and Clifton put it in \cite{HalvorsonClifton2002}, while particle detection events are ``illusory'', the quanta of a quantum field theory can be ``well-enough localized to give the appearance to us (finite observers) that they are strictly localized''. In this Section, I shall pursue the idea of ``approximately localized'' particles that supervene on the fundamental fields of QFT, leaving the confines of a looking for a strict particle ontology.
\subsection{Approximately local algebras}
In terms of abstract algebras, handling approximately local algebras is not straightforward. How can we interpret an observable that is ``approximately'' in a local algebra? What would that mean for nontrivial operations that annihilate the vacuum? In their \cite{HalvorsonClifton2002}, Halvorson and Clifton propose an algebra of observables `FAPP' localized in a region $O$: Given a algebra $\mathfrak{R}(O)$ local in the sense that it satisfies conditions (\ref{isotony}), (\ref{ttcovariance}) and (\ref{sm}). Then, an operator $A'$ is said to be \emph{approximately} local in $O$ if and only if it does not deviate from any $A \in \mathfrak{R}(O)$ by more than a fixed $\delta >0$. More precisely: the algebra of operators \emph{approximately} localized in $O$ is defined as 
\begin{equation} 
	\mathfrak{R}_{\delta}(O) = \{A': A\in \mathfrak{R}(O) \text{ and }\|A-A'\| < \delta \}
	\label{eq:approxalgebra}
\end{equation} 
where $\| \cdot \|$ is the operator norm, which, for every operator $C$ is given by the smallest number $\|C\|$ such that $\|C\|$ is the supremum of all $\|Cx\|$ for every unit vector in $\mathcal{H}$. Then, we would not be able to distinguish a measurement of $A'$ from a measurement of $A$, and therefore the algebra $\mathfrak{R}_{\delta}(O)$ would be, for all practical purposes (`FAPP'), localized in $O$. Unless $\delta = 0$, the Hilbert space $\mathcal{H}$ would then contain no vectors that are separating for the FAPP-localized algebra.
\subsection{A length scale for non-locality} \label{approxlocal}
In his \cite{Wallace2001}, Wallace introduces the notion of \emph{effective} localization for Hilbert space vectors, dependent on a lenght scale. Essentially, the idea is that we can consider a state in a Hilbert space $\mathcal{H}$ to be \emph{effectively} localized around a point in space if it decreases sufficiently fast as the distance from the point increases; this is gauged with respect to some scale $L$ which, in the KG case, turns out to be the Compton wavelength. A superposition of states effectively localized is then also effectively localized in effectively the same region, i.e. in a region that is congruent to the original one up to deviations of the order $L$. If these regions are large compared to the lengthscale $L$ and the Hilbert space is (approximately) preserved by the dynamics of our QFT, then one could define projectors onto the effectively localized states in $\mathcal{H}$. \\
In an algebraic context, this would correspond to defining an algebra of operators that are \emph{effectively} localized in the following sense: Given some GNS representation $(\pi, \mathcal{H}_{\pi})$ of a C*-algebra $\mathfrak{A}$ on the set of bounded linear operators $\mathfrak{B}(\mathcal{H}_{\pi})$, and a unit vector $x \in \mathcal{H}_{\pi}$. Then, a von Neumann subalgebra $\mathfrak{R} \subseteq \mathfrak{A}$ is said to be \emph{effectively} localized in a spatial set $G$ if any state $\Psi$ in the dense set $\{\pi(A)x: A \in \mathfrak{R}\}$ decreases rapidly relative to the appropriate length scale $L$ outside of $G$. It follows that, by the definition of the operator norm given above, one can always find a fixed upper bound $\delta$ in (\ref{eq:approxalgebra}) such that any algebra effectively localized relative to the length scale $L$ will be also be `FAPP'-localized. Conversely, any upper bound $\delta$ of a given `FAPP'-localized algebra gives rise to a length scale for effective localization. Again, this follows from the definition of the operator norm with respect to a suitable GNS unit vector. Thus, `FAPP' localized algebras give rise to effectively localized states in $\mathcal{H}$ and vice versa. \\
Recall that the operator $H = \sqrt{-\nabla^2 + m^2}$ has the positive spectrum $[m, \infty)$. In his paper \cite{Wallace2006}, Wallace shows that the lower bound of this spectrum, the inverse Compton length $L_c = \frac{1}{m}$, is the scale at which we can define \emph{approximate} locality of the KG states - the Compton length representing the wavelength of a field corresponding to the ``rest mass'' of the KG particle. It is at this point that the nature of QFT as a theory of fundamental fields is clearly visible: particular field configurations give rise to a particle-like phenomena - in the form of effectively localized particles of the KG theory, which are merely ``localized blobs of the field'', as Wallace aptly puts it in his \cite{Wallace2006}. From this point of view, our efforts in Sections \ref{locpartont}, \ref{naiveloc} and \ref{NWloc} can be regarded as an attempt to reinterpret a \emph{field} theory in such a way that it permits a fundamental particle interpretation. As we have seen, this not possible. Instead, particles can be interpreted as \emph{supervenient} on fields: Two physical systems with identical field configuration should yield an identical particle configuration. 

\section{Conclusion}
In this essay, I have given two localization schemes for the quanta of the Klein-Gordon field and shown how they fail to produce a particle ontology for the Klein-Gordon QFT. While the Newton-Wigner scheme avoids some of the problems facing standard localization, it still has a number of counterintuitive properties that do not warrant a fundamental interpretation of NW-localization for QFT. Finally, I explained how we can give approximately localized algebras and thus provide states that are effectively localized relative to the Compton wavelength, which can be interpreted as giving rise to particle observation events. Thus, free bosonic QFT is clearly a theory about \emph{fundamental} fields, and particles are only permitted as secondary physical quantities.

\begin{thebibliography}{99}
\bibliographystyle{alpha}
\bibitem{Arageorgis2003} A. Arageorgis et al. (2003): `Fulling Non-Uniqueness and the Unruh Effect: A Primer on Some Aspects of Quantum Field Theory', \emph{Philosophy of Science}, Vol. \textbf{70}, No. 1, pp. 164-202
\bibitem{Baker2008} D. Baker (2008): `Against Field Interpretations of Quantum Field Theory', \emph{British Journal for the Philosophy of Science} 2009, Vol. \textbf{60}, No. 3, pp. 585-609
\bibitem{BR1987} O. Bratteli and D. Robinson (1987): \emph{Operator Algebras and Quantum Statistical Mechanics}, Vol. \textbf{1}, Springer
\bibitem{BR1997} O. Bratteli and D. Robinson (1997): \emph{Operator Algebras and Quantum Statistical Mechanics}, Vol. \textbf{2}, Springer
\bibitem{Conway1990} J. Conway (1990): \emph{A course in Functional Analysis}, 2nd edition, Springer
\bibitem{Clifton1998} R. Clifton et al. (1998): `Superentangled States', \emph{Physical Review A}, Vol. \textbf{58}, No. 1, pp. 135-145
\bibitem{CliftonHalvorson2001} R. Clifton and H. Halvorson (2001): `Are Rindler Quanta Real?', \emph{The British Journal for the Philosophy of Science}, Vol. \textbf{52}(3), pp. 417-470
\bibitem{CliftonHalvorson2000} R. Clifton and H. Halvorson (2001): `Entanglement and Open Systems in Algebraic Quantum Field Theory', \emph{Studies In History and Philosophy of Science Part B: Studies In History and Philosophy of Modern Physics},
Vol. \textbf{32}, Issue 1, pp. 1-31 
\bibitem{FB1999}G. Fleming and J. Butterfield (1999): `Strange Positions', in J. Butterfield and C. Pagonis (eds.), \emph{From Physics to Philosophy}, Cambridge University Press
\bibitem{Fleming2000} G. Fleming (2000): `Reeh-Schlieder meets Newton-Wigner', \emph{Philosophy of Science}, Vol. \textbf{67}, Supplement. Proceedings of the 1998 Biennial Meetings of the Philosophy of Science Association. Part II: Symposia Papers, pp. S495-S515
\bibitem{Fleming2004} G. Fleming (2004): `Observations on Hyperplanes: II. Dynamical Variables and Localization Observables', \url{http://philsci-archive.pitt.edu/archive/00002085/}
\bibitem{Haag} R. Haag (1992): \emph{Local Quantum Physics}, Springer
\bibitem{Halvorson2001} H. Halvorson (2001): `Reeh-Schlieder defeats Newton-Wigner', \emph{Philosophy of Science}, Vol. \textbf{68}, No. 1, pp. 111-133
\bibitem{HalvorsonClifton2002}H. Halvorson and R. Clifton (2002): `No place for Particles in Relativistic Quantum Theories?', \emph{Philosophy of Science}, Vol. \textbf{69}, No. 1 pp. 1-28
\bibitem{HalvorsonMueger2007} H. Halvorson and M. Müger (2007): `Algebraic Quantum Field Theory', in J. Butterfield and J. Earman (eds.), \emph{Philosophy of Physics}, pp. 731-922, Elsevier
\bibitem{HalvorsonThesis} H. Halvorson (2001): `Locality, Localization, and the Particle Concept: Topics in the Foundations of Quantum Field Theory', PhD thesis, University of Pittsburgh
\bibitem{Kadison} R. Kadison (1970): `Some analytic methods in the theory of operator algebras', in C. Taam (ed.), \emph{Lectures of Modern Analysis and Applications}, Vol. \textbf{II.}, pp. 8-29, Springer
\bibitem{Kay1970} B. Kay (1979): `A uniqueness result in the Segal-Weinless approach to linear Bose fields', \emph{Journal of Mathematical Physics}, Vol. \textbf{20}, pp. 1712-1713
\bibitem{Malament1996} D. Malament (1996): `In Defense of Dogma: Why there cannot be a relativistic quantum mechanics of (localizable) particles', in R. Clifton (ed.), \emph{Perspectives on Quantum Reality}, Dordrecht:Kluwer, pp. 1-10
\bibitem{MB1994} F. Muller and J. Butterfield (1994): `Is Algebraic Lorentz-Covariant Quantum Field Theory Stochastic Einstein local?',  \emph{Philosophy of Science}, Vol. \textbf{61}, No. 3, pp. 457-474
\bibitem{NW1949} T. Newton and E. Wigner (1949): `Localized States for Elementary Systems', \emph{Reviews of Modern Physics}, Vol. \textbf{21}, pp. 400-406
\bibitem{Petz} D. Petz (1990): \emph{An Invitation to the Algebra of Canonical Commutation Relations}, Leuven University Press
\bibitem{Peskin} M. Peskin and D. Schroeder (1995): \emph{An Introduction to Quantum Field Theory}, Westview Press
\bibitem{Redhead1994} M. Redhead: `The Vacuum in Relativistic Quantum Field Theory', \emph{PSA: Proceedings of the Biennial Meeting of the Philosophy of Science Association} 1994, Vol. \textbf{2}: Symposia and Invited Papers, pp. 77-87
\bibitem{ReehSchlieder} H. Reeh and S. Schlieder (1961): `Bemerkungen zur Unitäräquivalenz von Lorentzinvarianten Feldern', \emph{Nuovo cimento (10)}, Vol. \textbf{22}, pp. 1051-1068
\bibitem{Schmid} R. Schmid (2006): `Infinite dimensional Hamiltonian systems', in \emph{Encyclopedia of Mathematical Physics} by J. Francoise, G. Naber, T. Tsung (eds.), Elsevier, Academic Press, 37-44.
\bibitem{SegalGoodman} I. Segal and R. Goodman (1965): `Anti-locality of certain Lorentz-invariant operators', \emph{Journal of Mathematics and Mechanics}, Vol. \textbf{14}, pp. 629-638
\bibitem{Wald1994} R. Wald (1994): \emph{Quantum Field Theory in Curved Spacetime and Black Hole Thermodynamics}, University of Chicago Press
\bibitem{Wallace2001} D. Wallace (2001): `Emergence of Particles from Bosonic Quantum Field Theory', \url{http://arxiv.org/abs/quant-ph/0112149}
\bibitem{Wallace2006} D. Wallace (2006): `In Defense of Na\"{i}veté: The conceptual status of Lagrangian quantum field theory', \emph{Synthese}, 	Volume \emph{151}, Number 1, pp. 33-80
\bibitem{Weinstein} A. Weinstein (1969): `Symplectic Structures on Banach Manifolds', \emph{Bulletin of the American Mathematical Society}, Vol. \textbf{75}, pp. 1040-1041
\end{thebibliography}
\end{document}